\documentclass[11pt,letterpaper]{amsart}
\usepackage{amsfonts}
\usepackage[english]{babel}
\usepackage{amssymb,amsthm,amsmath}
\usepackage[numbers,sort&compress]{natbib}
\usepackage{color}
\usepackage{graphicx}
\usepackage[colorlinks,
            linkcolor=blue,
            anchorcolor=blue,
            citecolor=blue
            ]{hyperref}

\usepackage{fancyhdr}
\newcommand\shortitle{Area Law for the Entanglement Entropy, Nonrandom Ergodic Field}
\numberwithin{equation}{section}
\def\R{\mathbb{R}}

\def\Z{\mathbb{Z}}

\def\bm{\begin{matrix}}
\def\em{\end{matrix}}

\theoremstyle{plain}
\newtheorem{theorem}{Theorem}
\theoremstyle{plain}
\newtheorem{proposition}{Proposition}[section]
\newtheorem{defin}[proposition]{Definition}
\newtheorem{definition}[proposition]{Definition}
\newtheorem{lemma}[proposition]{Lemma}
\newtheorem{rmk}[proposition]{Remark}

\newtheorem*{claim*}{Claim}
\newtheorem{criterion}{Criterion}

\newtheoremstyle{named}{}{}{\itshape}
{}{\bfseries}{.}{.5em}{\thmnote{#3}#1}
\theoremstyle{named}

\begin{document}

\title{Area Law for the entanglement entropy \\
of free fermions \\
in nonrandom ergodic field}
\author{Leonid Pastur}
\address{Department of mathematics, King's College London, Strand, London
WC2R 2LS}
\email{leonid.pastur@kcl.ac.uk}
\address{Theoretical Department, B.Verkin Institute for Low Temperature
Physics and Engineering, 47 Nauky Avenue, Kharkiv, 61103, Ukraine}
\email{pastur@ilt.kharkov.ua }
\author{Mira Shamis}
\address{School of Mathematical Sciences, Queen Mary University of London,
Mile End Road, London E1 4NS, England}
\email{m.shamis@qmul.ac.uk}
\address{School of Mathematical Sciences, Holon Institute of Technology, 52 Golomb St.,
P.O.B. 305, Holon 5810201, Israel}
\email{miras@hit.ac.il}
\date{\today}

\begin{abstract}

This paper deals with the asymptotic behaviour of a widely used correlation characteristic in large quantum systems.
The correlations are known as quantum entanglement, the characteristic is called entanglement entropy, and the system is an ideal gas of spinless lattice fermions. The system is determined by its one-body Hamiltonian. It is shown in EPS [18] that if the Hamiltonian is an ergodic finite difference operator with an exponentially decaying spectral projection, then the asymptotic form of the entanglement entropy is the so-called area law. However, the only class of one-body Hamiltonians for which this spectral condition was verified is that consisting of discrete Schrödinger operators with random potential. In this paper, we prove the validity of the area law for several classes of Schrödinger operators whose potentials are ergodic but not random. We begin with quasiperiodic and limit-periodic operators and then move on to the interesting and highly non-trivial case of potentials generated by subshifts of finite type. These arose in the theory of dynamical systems when studying non-random chaotic phenomena. The corresponding asymptotic study requires quite an involved spectral analysis. Consequently, the majority of the paper is devoted to the proof and application of a variety of spectral properties of the operators in question, in particular we prove uniform localisation of the ejgenfunctions for the Maryland model and the exponential decay of the eihgenfunction correlator for a variety of models . We believe that these properties are of considerable independent interest.

\end{abstract}

\maketitle



\section{Introduction}

Quantum entanglement is one of the basic properties of quantum systems.
Introduced by Einstein, Rosen, and Podolsky in 1935 to prove the
incompleteness of quantum mechanics and immediately identified by
Schr\"odinger as an important form of quantum correlations, quantum
entanglement is now the subject of active research in both fundamental and
applied fields, see e.g. review works \cite{Al-Co:21, Ca-Co:11, Ca-Hu:09, Du:22,
La:15, Sc:19, Wi:18} and references therein.

One of the widely used quantitative characteristics of quantum entanglement
between two subsystems of a quantum system is the entanglement entropy. Much
its relevant studies, which are also of importance for quantum statistical
mechanics, quantum gravity, and quantum computing deal with the analysis of
the asymptotic behavior of the entanglement entropy of a sufficiently large
(mesoscopic) subsystem confined to a finite block $\Lambda $ of a
macroscopically large quantum system occupying the whole $\mathbb{Z}^{d}$ or 
$\mathbb{R}^{d}$.

As a result of numerous theoretical, experimental, and numerical works of
recent decades, it was found on various level of rigor that the following
two basic asymptotic forms of the entanglement entropy $S_{\Lambda }$ of a
block $\Lambda $ of linear size $L$ are valid for a wide class of quantum
systems. 

\smallskip - Area Law:
\begin{equation}
S_{\Lambda }=C^{\prime }L^{d-1}(1+o(1)),\;L\rightarrow \infty,  \label{eq:al}
\end{equation}
if the system ground state is not critical (no quantum phase transition)
or/and if there is a spectral gap between the ground state and the rest of
the spectrum;

\smallskip - Enhanced Area Law: 
\begin{equation}
S_{\Lambda}=C^{\prime \prime }L^{d-1}\log L(1+o(1)),\;L\rightarrow \infty,
\label{eq:eal}
\end{equation}
if the system ground state is critical (a quantum phase transition is the
case).

For the sake of completeness, we will also mention one more asymptotic form
of entropy, which is now known as the

\smallskip - Volume law 
\begin{equation*} 
S_{\Lambda }=C^{\prime \prime\prime }L^{d}(1+o(1)),\;L\rightarrow \infty,
\end{equation*}
which dates back to the origin of quantum statistical mechanics and is the
case if the system is either in a mixed state, say, the Gibbs state of
non-zero temperature, or in a pure but sufficiently highly excited state.

The rigorous proof of these asymptotic formulae, especially $\eqref{eq:al}$
and $\eqref{eq:eal}$, proved to be quite nontrivial in the general case of
multidimensional interacting quantum systems. This is why considerable
attention has been paid to a simple yet nontrivial system of free fermions.
The system, known since the late $1920$s as a fairly adequate model of
metals, has a number of properties that make it quite interesting in modern
physics and related fields. The system is completely characterized by its
one-body Hamiltonian $H$, a selfadjoint operator in $\ell^2(\mathbb{Z}^{d})$
or $L^2 (\mathbb{R}^{d})$.

We will confine ourselves to the lattice case, i.e., to $\mathbb{Z}^{d}$ as
the space where the system lives. 
Then the one-body Hamiltonian is a selfadjoint operator in $\ell ^{2}(%
\mathbb{Z}^{d})$ 
\begin{equation}
H=\{H(\mathbf{m},\mathbf{n})\}_{\mathbf{m},\mathbf{n}\in \mathbb{Z}^{d}},\;%
\overline{H(\mathbf{m},\mathbf{n})}=H(\mathbf{n},\mathbf{m}),  \label{eq:H}
\end{equation}%
and the one of the most important cases is where $H$ is a discrete Schr\"{o}%
dinger operator 
\begin{equation}
\begin{split}
(H\psi )(\mathbf{n})& =(\Delta \psi )(\mathbf{n})+(V\psi )(\mathbf{n}) \\
& =\sum_{|\mathbf{n}-\mathbf{m}|=1}\psi (\mathbf{m})+V(\mathbf{n})\psi (%
\mathbf{n}),\;\mathbf{n, m}\in \mathbb{Z}^{d},
\end{split}
\label{eq:schrodinger}
\end{equation}%
which is an archetype model of the field.

Denote by $\sigma (H)$ the spectrum of $H$ and by $\varepsilon _{-}$
and $\varepsilon _{+}$ the extreme endpoints of $\sigma (H)$, so that 
\begin{equation}
\sigma (H)\subset \lbrack \varepsilon _{-},\varepsilon _{+}]=:I(H).
\label{eq:kap}
\end{equation}

Let $\mathcal{E}_{H}(d\lambda )$ be the resolution of identity of $H$.
Introduce the Fermi projection 
\begin{equation}
P(\varepsilon _{F})=\mathcal{E}_{H}((\varepsilon _{-},\varepsilon _{F}]) =
\chi_{(\varepsilon_{-}, \varepsilon_F]}(H),  \label{eq:FP}
\end{equation}%
where $\varepsilon _{F}$ is a parameter of the system, called the Fermi
energy (the ground state energy of free fermions) {and $\chi_{(\cdot, \cdot]}(\cdot)$ is the indicator function}.

Then the entanglement entropy of the finite block $\Lambda \subset \mathbb{Z}%
^{d}$ of free fermions is \cite{Ei-Co:11,Pe-Ei:09} 
\begin{equation}
S_{\Lambda }(\varepsilon _{F})=\mathrm{Tr}_{\Lambda }h(P_{\Lambda
}(\varepsilon _{F})),  \label{eq:ee}
\end{equation}%
where $\mathrm{Tr}_{\Lambda }$ denotes the (restricted) trace in $\ell
^{2}(\Lambda )$, 
\begin{equation}
P_{\Lambda }(\varepsilon _{F})=\chi _{\Lambda }P(\varepsilon _{F})\chi
_{\Lambda }\equiv P(\varepsilon _{F})|_{\Lambda }  \label{eq:pla}
\end{equation}%
is the restriction of $P(\varepsilon _{F})$ of $\eqref{eq:FP}$ to $%
\ell ^{2}(\Lambda )$, 
\begin{equation}
\chi _{\Lambda }:\ell ^{2}(\mathbb{Z}^{d})\rightarrow \ell ^{2}(\Lambda )
\label{eq:chi}
\end{equation}%
is the coordinate projection, and 
\begin{equation}
h(x)=-x\log x-(1-x)\log (1-x),\;x\in \lbrack 0,1]  \label{eq:h}
\end{equation}%
is the Shannon binary entropy.

Note that if $\varepsilon_F\notin I(H)$ of $\eqref{eq:kap}$, then $%
S_\Lambda(\varepsilon_F) = 0$. Thus, from now forward, we are assuming that 
\begin{equation}  \label{eq:eI}
\varepsilon_F\in I(H),
\end{equation}
i.e., $\varepsilon_F$ belongs either to the spectrum of $H$ or to its
internal gap.

To expect the regular asymptotic behavior of the entanglement entropy, one
has to assume a certain ``homogeneity" of the one-body Hamiltonian. This is
why we will consider the class of the so-called ergodic operators that
possess this important property. The class is defined as follows.

Let $(\Omega ,\mathcal{F},\mathbf{P})$ be a probability space equipped with
a measure-preserving $d$-dimensional group of ergodic transformations $\{T^{%
\mathbf{n}}:\Omega \rightarrow \Omega \}_{\mathbf{n}\in \mathbb{Z}^{d}}$.
Then a $d$-dimensional ergodic operator acting in $\ell ^{2}(\mathbb{Z}^{d})$
is an operator-valued random variable assuming values in selfadjoint
operators and such that (see \eqref{eq:H}) 
\begin{equation}
H_{\omega }=\{H(\omega ,\mathbf{m},\mathbf{n})\}_{\mathbf{m},\mathbf{n}\in 
\mathbb{Z}^{d}},\;H(T^{\mathbf{k}}\omega ,\mathbf{m},\mathbf{n})=H(\omega ,%
\mathbf{m}+\mathbf{k},\mathbf{n}+\mathbf{k}),  \label{eq:herg}
\end{equation}%
see e.g. \cite{AW,Pa-Fi:92} for spectral theory of this class of operators.

The simplest but quite important subclass of ergodic operators in $\ell ^{2}(%
\mathbb{Z}^{d})$ are discrete convolution operators, where $\Omega
=\{\emptyset \}$ and 
\begin{equation}
H=\{H(\mathbf{m}-\mathbf{n})\}_{\mathbf{m},\mathbf{n}\in \mathbb{Z}%
^{d}},\sum_{\mathbf{m}{\in \mathbb{Z}^{d}}}|H(\mathbf{m})|<\infty .
\label{eq:con}
\end{equation}%
It was shown in a series of works (see \cite{Le-Co:14, Le-Co:15, Le-Co:17,
Pf-So:18, So:17} and references therein), that for a wide class of
translation invariant pseudodifferential operators (including \eqref{eq:con}%
) the corresponding entanglement entropy obeys the 
{Enhanced Area Law $\eqref{eq:eal}$ if the Fermi energy $\varepsilon _{F}$
of $\eqref{eq:FP}-\eqref{eq:pla}$ belongs to the (absolutely continuous)
spectrum $\sigma (H)=\sigma _{ac}(H)$ of $H$. 
}

This should be compared with the \textquotedblleft opposite" case of random
ergodic operators, where $\Omega =\mathbb{R}^{\mathbb{Z}}$, (e.g. the
discrete Schr\"{o}dinger operators with an independent identically
distributed (i.i.d.) potential \cite{AW,Pa-Fi:92}), where the Area Law \eqref{eq:al} for
the expectation $\mathbf{E}\{S_{\Lambda}(\varepsilon_F)\}$ and even for the
overwhelming majority of realization holds for $\varepsilon _{F}$ belonging
to a certain part of the pure point spectrum $\sigma _{pp}(H) \subset
\sigma(H)$, which is due to the so-called \textquotedblleft strong" Anderson
localisation \cite{EPS,Pa-Sl:18}. Recall that the spectrum of an ergodic
operator and its all components are deterministic, i.e., are independent of $%
\omega$ with probability 1.

An interesting observation that follows from the above is that the
asymptotic behavior of the entanglement entropy is closely related to the
spectral type of the corresponding one-body operator. Namely, for
convolution operators (translation invariant case), where the spectrum is
purely absolutely continuous and the generalized eigenfunctions are plane
waves, we have the Enhanced Area Law $\eqref{eq:eal}$, while for random
operators (disordered case), where the spectrum is pure point and
eigenfunctions are exponentially decaying, we have the Area Law $\eqref{eq:al}$. Here it is appropriate to recall a result of \cite{Br-Ho:15}%
, according to which the exponential decay of correlation functions in
one-dimensional quantum systems implies the Area Law for the entanglement
entropy of their states {
under certain additional conditions. However, this observation cannot be true as formulated, since, 
for example, there exist random Schr\"odinger operators with a purely point spectrum that exhibit the Enhanced Area Law \eqref{eq:eal} for isolated 
values of the Fermi energy \cite{Mu-Co:20}. This is because the asymptotic form of the entanglement entropy is determined by 
the behaviour at $\varepsilon_F$ (and, possibly, in $\Lambda$-infinitesimal neighbourhoods of $\varepsilon_F$) of the entries $\{P({\mathbf{m,n}})\}_{\mathbf{m,n} \in \mathbb{Z^d}}$ of the Fermi projection 
as $|\mathbf{m-n}| \to \infty$, 
while the spectral type is determined by the spatial behaviour of the matrix
on $\Lambda$-independent intervals (see, e.g. \eqref{eq:FP1-con}, \eqref{eq:exp-decay}).  A similar ``sensitivity" is also known for some transport characteristics of 
corresponding disordered systems with the same single-particle Hamiltonian 
(moments of the position operator, d.c. conductivity, etc.), see \cite{Go-Me:78} and \cite{Ji-Sc:07} and references therein.

Nevertheless, the} observation allows us to expect that the Area Law holds for the classes
of ergodic operators other than random ones, provided that they also exhibit sufficiently
``strong" Anderson localisation. In particular, these are finite difference
operators with quasi-periodic potentials, defined on an orbit of an
irrational winding on the torus, or with potentials that are functions
defined on an orbit of a more complex finite-dimensional dynamical systems,
including the doubling map and the famous Arnold's cat map \cite{Br-St:02,
Adler}.

The goal of this paper is to prove the Area Law for these classes of ergodic
operators. The paper is organised as follows. In Section $\ref{sec:res_comm}$
we present our results on the validity of the Area Law for
free lattice fermions whose one-body Hamiltonian is a finite-difference
operator in $\ell ^{2}(\mathbb{Z}^{d})$ with dynamically generated
coefficients. Specifically, these are self-adjoint ergodic operators that are
the sums of a 
convolution operator \eqref{eq:con} (the discrete Laplacian in the case of
Schr\"{o}dinger operator \eqref{eq:schrodinger}) and a dynamically generated
potential, 
see \eqref{eq:lro} --
\eqref{eq:W} and \eqref{eq:ergodic-potential-1} below. Theorems~\ref{th:ap-d}
and ~\ref{th:ap-1d} treat multidimensional and one-dimensional
quasiperiodic and limit-periodic cases.
Theorem~\ref{subshit-ee} treats the case, where the one-body Hamiltonian is  the one-dimensional Schr\"{o}dinger operator with potentials generated by a subshift of finite type, a
complex dynamical system that is much more complicated than the irrational
shift on $\mathbb{T} = \R/\Z$, which generates quasi-periodic potentials.


%

All of these operators have a purely point component of spectrum. However, our proof, based on the results of \cite{EPS}, makes significant use of some form of exponential bounds on the spectral (Fermi) projection of the operator in question, i.e., a considerably stronger property. In certain cases this property is partially established. 
There are, however, several important cases where the existing spectral theory has to be considerably developed. The first is the case of the archetype Maryland model, both multi-dimensional and 
one-dimensional, where we  have to prove the property of uniform localisation of eigenfunctions, presented  in Theorem~\ref{th:ule-d-maryland}. The second important case is the Schr\"{o}dinger operator whose potentials  generated by the subshift of finite type. Here we have to prove the exponential decay of the eigenfunction correlator, presented in Theorem~\ref{th:subshift-ee}. This includes a useful criterion for the exponential decay of the eigenfunction correlator for ergodic bounded one-dimensional discrete Schr\"{o}dinger operators in terms of
control over the ``bad" sets of spectral parameters of the operator restriction to large but distant boxes, see Lemma~\ref{lem:bad-res}. We believe that these spectral results are of independent interest.

\section{Results and Comments}

\label{sec:res_comm} In this section we formulate our main results and make
certain comments on their meaning and proof.

All results correspond to the one-body operator (cf. (\ref{eq:H}))%
\begin{align}
& H_{\omega }=W+V_{\omega }=\{H(\omega ,\mathbf{m},\mathbf{n})\}_{\mathbf{m},%
\mathbf{n}\in \mathbb{Z}^{d}},  \notag \\
& \;H(\omega ,\mathbf{m},\mathbf{n})=W(\mathbf{m-n})+V(\omega ,\mathbf{n}%
)\delta (\mathbf{m-n}),  \label{eq:lro}
\end{align}%
where 
\begin{align}
& W(-\mathbf{n})=\overline{W(\mathbf{n})},\;|W(\mathbf{n})|\leq We^{-\rho |%
\mathbf{n}|},  \notag \\
& \;W<\infty ,\;\rho >0,\;|\mathbf{n}|=|n_{1}|+\cdots +|n_{d}|,  \label{eq:W}
\end{align}%
and its particular case -- Schr\"{o}dinger operator \eqref{eq:schrodinger},
where $W$ is the $d$-dimensional discrete Laplacian
\begin{equation}\label{eq:WL}
W(\mathbf{n})=\delta _{\mathbf{n}+\mathbf{e}_{j}}+\delta _{\mathbf{n}-%
\mathbf{e}_{j}},\;\mathbf{n}\in \mathbb{Z}^{d},\;j=1,\dots ,d,  
\end{equation}
and $\{\mathbf{e}_{j}\}_{j=1}^{d}$ is the canonical basis in $\mathbb{Z}^{d}$.

Since $H_{\omega }$ is ergodic (see \eqref{eq:H}), its potential $V_{\omega
}=\{V(\mathbf{n},\omega \}_{\mathbf{n}\in \mathbb{Z}^{d}}$ in \eqref{eq:lro} 
is an ergodic field in $\mathbb{Z}^{d}$. Namely, it is given by a measurable
(sample) function $v:\Omega \rightarrow \mathbb{R}$ on a probability space $%
(\Omega ,\mathcal{F},\mathbf{P})$ equipped with a measure preserving and
ergodic group of automorphisms $\{T^{\mathbf{n}}:\,\Omega \rightarrow \Omega
\}_{\mathbf{n}\in \mathbb{Z}^{d}}$, so that 
\begin{equation}
V_{\omega }=\{V(\omega ,\mathbf{n})\}_{\mathbf{n}\in \mathbb{Z}%
^{d}},\,V(\omega ,\mathbf{n})=v(T^{\mathbf{n}}\omega ),\,\mathbf{n}\in 
\mathbb{Z}^{d},\,\omega \in \Omega .  \label{eq:ergodic-potential-1}
\end{equation}%
We will also assume in what follows that the block in \eqref{eq:ee} -- %
\eqref{eq:chi} is a cube 
\begin{equation*}
\Lambda =[-M,M]^{d}\subset \mathbb{Z}^{d},\;2M+1=L. 
\end{equation*}

%
%


\subsection{Results}

\label{ssec:results}

We present here the results on the validity of the Area Law \eqref{eq:al}
for several interesting classes of ergodic operators. Recall that the only
class of these operators for which 
the Area Law has been rigorously established is of the Schr\"odinger
operators with independent identically distributed random potentials whose
probability law possesses a certain amount of smoothness \cite{EPS}.

\smallskip

We will begin with results valid in any dimension $d\geq 1$. {We denote by $\mathbb T^d = \R^d/ \Z^d$ the $d$-dimensional torus for any $d\geq 1$, and by $\langle \cdot, \cdot\rangle$ the scalar product in $\Z^d$.} 

\smallskip
\begin{theorem}
\label{th:ap-d} The expectation $\mathbf{E} \{S_\Lambda (\varepsilon_F)\}$
of the entanglement entropy \eqref{eq:ee} -- \eqref{eq:h} of free lattice
fermions obeys the Area Law \eqref{eq:al} for the following ergodic one-body
Hamiltonians acting in $\ell^2(\mathbb{Z}^d), \; d \ge 1 $:

\begin{itemize}
\item[(i)]\label{item-d-maryland} The multidimensional Maryland model: the operator \eqref{eq:lro}
-- \eqref{eq:W} with the potential 
\begin{equation}
V(\omega ,\mathbf{n})=g\,\tan \pi(\omega +\langle \mathbf{n},\mathbf{\alpha }
\rangle ),\,\mathbf{n}\in \mathbb{Z}^{d},\;\omega \in \mathbb{T} =\Omega
,\;\,g\neq 0,  \label{eq:maryland-d}
\end{equation}%
where $\mathbf{\alpha }\in \mathbb{T}^{d}$ satisfies the multidimensional
Diophantine condition 
\begin{equation}
||\langle \mathbf{n},\mathbf{\alpha }\rangle ||_{d}\geq c|\mathbf{n}|^{-\tau
},\;c>0,\tau >0,\;\mathbf{n}\in \mathbb{Z}^{d}\setminus \{0\},
\label{eq:Dob}
\end{equation}%
with $\Vert \langle \mathbf{n},\mathbf{\alpha }\rangle \Vert _{d}=\mathrm{%
dist}(\langle \mathbf{n},\mathbf{\alpha }\rangle ,\mathbb{Z})$ and the Fermi
energy $\varepsilon_F\in\R$
$($cf. Theorem \ref{th:ap-1d} $($i$)$ about the $1$-d case below$)$.

\medskip

\item[(ii)] The Schr\"{o}dinger operators $\eqref{eq:schrodinger}$ with the potential 
\begin{equation}
V(\omega ,\mathbf{n})=g\,v(\omega +\langle \mathbf{n},\alpha \rangle ),\,
\mathbf{n}\in \mathbb{Z}^{d},\;\omega \in \mathbb{T},  \label{eq:lpp-d}
\end{equation}
where $g>0$ is sufficiently large, $v:\mathbb{R}\rightarrow \mathbb{R}$ is $%
1 $-periodic function with the $\xi $-H\"{o}lder monotone property 
\begin{equation}
v(y)-v(x)\geq (y-x)^{\xi },\,\xi \geq 1,\,0\leq x\leq y<1,
\label{eq:monotone-d1}
\end{equation}
$\alpha \in \mathbb{T}^{d}$ satisfies a weak Diophantine condition $($cf. \eqref{eq:Dob} and \eqref{eq:DC_d}$)$
\begin{equation}
\Omega _{\rho ,\mu }=\left\{ \alpha \in \mathbb{T}^{d}\,:\,\Vert \langle 
\mathbf{n},\alpha \rangle \Vert _{d}\geq \exp{\big\{-\rho |\mathbf{n}|^{\frac{1}{1+\mu }}}\big\},\mathbf{n}\neq \mathbf{0} \right\} ,  \label{eq:weakly-Dio}
\end{equation}%
with $\rho ,\mu >0$, $\Vert \langle \mathbf{n},\mathbf{\alpha }\rangle \Vert
_{d}=\mathrm{dist}(\langle \mathbf{n},\mathbf{\alpha }\rangle ,\mathbb{Z})$
and the Fermi energy is%
\begin{equation*}
\varepsilon _{F}\in I(H),
\end{equation*}%
where $I(H)$ is defined in \eqref{eq:kap} $($cf. Theorem~\ref{th:ap-1d} $(iii))$.

\medskip

\item[(iii)] The operator \eqref{eq:lro} with $\{W(\mathbf{n})\}_{\mathbf{n}%
\in \mathbb{Z}^{d}}$ satisfying \eqref{eq:W}, and with the multidimensional
analog%
\begin{equation}
V(\omega ,\mathbf{n})=g\cos 2\pi(\omega +\langle \mathbf{n},\mathbf{\alpha }%
\rangle ),\;\ \mathbf{n}\in \mathbb{Z}^{d},\;\omega \in \mathbb{T},
\label{eq:amo-d}
\end{equation}%
of the almost Mathieu potential $\eqref{eq:amo}$, where $g$ is large enough, 
$\mathbf{\alpha }\in DC_{d}$ with 
\begin{equation}
DC_{d}=\bigcup_{\kappa>0,\tau >d-1}\Big\{\mathbf{\alpha }\in 
\mathbb{T}^{d}:\inf_{j\in \mathbb{Z}}|\langle \mathbf{n},\mathbf{\alpha }\rangle -j|> {\kappa}{|\mathbf{n}|^{-\tau }},\,\,\mathbf{n}\in 
\mathbb{Z}^{d}\setminus \{0\}\Big\},  \label{eq:DC_d}
\end{equation}%
$($cf. \eqref{eq:Dob}$)$, and with the Fermi energy $\varepsilon _{F}\in
I(H) $, where $I(H)$ is defined in \eqref{eq:kap}, see also \eqref{eq:eI}.

{

\item[(iv)] 
The Schr\"odinger operators $\eqref{eq:schrodinger}$ with a limit-periodic potentials 
\begin{equation}\label{eq:lp}
V_\omega(\mathbf n) = f(T^{\mathbf n}\omega), \; \mathbf n\in\Z^d, \;  \omega\in\Omega,
\end{equation}
where $\Omega$ is a Cantor group that admits a minimal $\Z^d$ action $T$ by translations $($see Definition~\ref{def:cantor-group}$)$, and $f\in C(\Omega, \R)$ $($the existence of such $\Omega$ and $f$ is proved in \cite{Da-Ga:13}, see Proposition~\ref{prop:da-ga} below$)$, {
where $C(\Omega, \R)$ is the space of continuous functions $f:\Omega\rightarrow\R$}, and the Fermi energy $\varepsilon_F\in I(H)$,  where $I(H)$ is defined in \eqref{eq:kap}, see also \eqref{eq:eI}.
}
\end{itemize}
\end{theorem}

{ 
\begin{rmk}\label{rmk:da-ga:11} $(i)$ There is a one-dimensional version of Theorem~\ref{th:ap-d} $(iv)$. It was considered  in \cite{Da-Ga:11} where an explicit condition on the Cantor group $\Omega$ is formulated.

$(ii)$  Examples of the corresponding ergodic limit-periodic potentials can be obtained by using the techniques developed in \cite{A,Da-Ga:11,Da-Ga:13, G} and limit-periodic potentials found in \cite{Poschel}. 


\end{rmk}
}

\medskip

We will present now certain one-dimensional results on the validity of the
Area Law, that are either more general or stronger  than their
multidimensional counterparts given in the above theorem. The results are mostly
for the one-dimensional Schr\"odinger operator (cf. \eqref{eq:schrodinger} and 
\eqref{eq:lro}) 
\begin{equation}
(H\psi )(n) =\psi(n+1) +\psi(n-1)+(V\psi )(n), \; n\in \mathbb{Z},
\label{eq:schrodinger1}
\end{equation}
with quasi-periodic and limit-periodic potentials.

We have already seen that the spectral properties of quasiperiodic operators
depend strongly on the arithmetic properties of the potential frequencies.
Here is a natural quantifier of their ``irrationality"
\begin{equation}  \label{eq:beta}
\beta(\alpha)=\lim \sup _{n \to \infty}\frac{-\log ||n \alpha||_\mathbb{T}}{n%
}=\lim \sup_{n \to \infty}\frac{\log q_{n+1}}{q_{n}},
\end{equation}
where $\{q_n\}_{n \in \mathbb{N}}$ are denominators of the continuous
fraction expansion of $\alpha=\lim_{n \to \infty} p_n/q_n$ and $||\cdot||_%
\mathbb{T}=\mathrm{dist_H}(\cdot,\mathbb{Z})$ {
(the Hausdorff distance)}.

The frequencies $\alpha$ for which $\beta(\alpha)=0$ are ``generalized"
Diophantine, cf. \eqref{eq:Dob}, and those with $0 <\beta(\alpha) \leq
\infty $ are Liouvillian.

\smallskip

\begin{theorem}
\label{th:ap-1d} The expectation $\mathbf{E} \{S_\Lambda (\varepsilon_F)\}$
of entanglement entropy $\eqref{eq:ee}$ -- $\eqref{eq:h}$ of free lattice
fermions with $d=1$ obeys the Area Law $\eqref{eq:al}$ for the
following one-body Hamiltonians:

\smallskip

\begin{itemize}
\item[(i)] The one-dimensional Maryland model given by \eqref{eq:lro} -- \eqref{eq:WL}  for $d=1$, i.e. the one-dimensional discrete
Schr\"{o}dinger operator $\eqref{eq:schrodinger1}$ with the potential $($cf. \eqref{eq:maryland-d}$)$ 
\begin{equation}
V(\omega ,n)=g\,\tan \pi(\omega +n\alpha ),\;\omega \in \mathbb{T},\,\alpha \in 
\mathbb{T},\,g\neq 0,  \label{eq:maryland}
\end{equation}%
and with the Fermi energy $\eqref{eq:FP}$ $\varepsilon _{F}\in \sigma
_{-}(H_{\omega }):=(- \infty, -\varepsilon_0)$, where $\varepsilon_{0} \ge 0$ solves the equation  
\begin{align}
\beta(\alpha)=\gamma(\varepsilon_0,g),
\label{eq:gabe}
\end{align}
with $\beta (\alpha )$  given in $\eqref{eq:beta}$ and $\gamma (\lambda
,g)\geq \gamma (0,g)>0$ being the Lyapunov exponent of the corresponding
finite-difference equation of the second order \eqref{eq:maryeq} $($see 
\eqref{eq:lyap-exp} for the general definition of the Lyapunov exponent, and \eqref{eq:le-mary}--\eqref{eq:le-mary3} for the form
of $\gamma (\lambda ,g)$ the one-dimensional
Maryland model$)$.  If $\gamma (0,g)\ge\beta (\alpha )$, then we have $\varepsilon_F\in\R$.



\smallskip

\item[(ii)] The one-dimensional Schr\"{o}dinger operator $%
\eqref{eq:schrodinger1}$ with potential 
\begin{equation}
V(\omega ,n)=gv(\omega +n\alpha ),\;\omega \in \mathbb{T},\,\alpha \in 
\mathbb{T},  \label{eq:lpp}
\end{equation}%
where $v:\mathbb{R}\rightarrow \mathbb{R}$ is $1$-periodic, continuous on $%
[0,1)$, $v(0)=0,v(1-0)=1$, and Lipschitz monotone $($cf. \eqref{eq:monotone-d1}$%
)$ 
\begin{eqnarray}
&&a_{-}(y-x)\leq v(y)-v(x)\leq a_{+}(y-x),\,\,  \notag \\
&&\hspace{-0cm}0\leq x\leq y<1,\,a_{\pm }>0,  \label{eq:lip}
\end{eqnarray}%
$\alpha $ satisfies the Diophantine condition $($cf. \eqref{eq:Dob}$)$
\begin{equation*}
\Vert n\alpha \Vert >C|n|^{-\tau },\;\,n\in \mathbb{Z},\;\tau >0,
\end{equation*}
where $\Vert \cdot \Vert =\min (\{\cdot \},\{1-\cdot \})$, $|g|$ is large
enough, and 
\begin{equation*}
\varepsilon _{F}\in I(H)\subset \lbrack -2,2+|g|],
\end{equation*}%
where $I(H)$ is defined in \eqref{eq:kap}, see also \eqref{eq:eI}.

\medskip

\item[(iii)] The supercritical almost Mathieu operator: the one-dimensional
discrete Schr\"{o}dinger operator $\eqref{eq:schrodinger1}$ with potential 
\begin{equation}
V(\omega ,n)=2g\cos 2\pi(\omega +n\alpha ),\;\omega \in \mathbb{T},\;|g|>1,
\label{eq:amo}
\end{equation}%
where $\alpha \in \mathbb{T}$ is such that 
$0\leq \beta (\alpha )<\beta $ for some $\beta >0$, $\beta (\alpha )$ is
defined by $\eqref{eq:beta}$, and the Fermi energy satisfies
\begin{equation*}
\varepsilon _{F}\in I(H)\subset \lbrack -(2+2|g|),(2+2|g|)]
\end{equation*}%
where $I(H)$ is defined in \eqref{eq:kap}, see also $\eqref{eq:eI}$.
\smallskip
\end{itemize}
\end{theorem}


{
The crucial ingredient of the proof of Theorem~\ref{th:ap-d}[(i)] and Theorem~\ref{th:ap-1d}[(i)], is the following assertion of independent interest, according to which the eigenfunctions of the multi- and one-dimensional Maryland model are uniformly localised. This result provides a first example of uniformly localised eigenfunctions for an \textit{unbounded} operator on the whole (infinite) spectrum. In the one-dimensional case the spectrum of the corresponding operator (the one-dimensional
Schr\"odinger operator \eqref{eq:schrodinger1} with potential \eqref{eq:maryland}) is described by similar formulae.
However, since in this case the operator $W$ in 
\eqref{eq:lro} -- \eqref{eq:W} is the one-dimensional discrete Laplacian, an
important part of the proof of the multidimensional estimate \eqref{eq:th-ulema} can be made explicit.
thereby leading to a stronger one-dimensional version.
\begin{theorem}\label{th:ule-d-maryland}
\begin{itemize}
\item[(i)] Let $d > 1$, and let $H_\omega$ be the multidimensional Maryland model, the operator \eqref{eq:lro}
-- \eqref{eq:W} with the potential  \eqref{eq:maryland-d}, where $\alpha$ satisfies the multidimensional Diophantine condition \eqref{eq:Dob}. Then the spectrum of $H_\omega$ occupies the whole $\mathbb{R}$, is pure point and of multiplicity 1 $($see e.g. \cite[Sections 18.A-18.B]{Pa-Fi:92}$)$. If $\{\lambda _{\mathbf{l}}(\omega )\}_{\mathbf{l}\in \mathbb{Z}^{d}}$ are its eigenvalues and  $\{\psi _{\mathbf{l}}(\omega )\}_{\mathbf{l}\in \mathbb{Z}^{d}}$, $\psi _{\mathbf{l}}(\omega )=\{\psi _{
\mathbf{l}}(\omega ,\mathbf{n})\}_{\mathbf{n}\in \mathbb{Z}^{d}}$ are the corresponding eigenfunctions of $H_\omega$, then there exist  $C<\infty $ and $c>0$ independent of $\lambda \in \mathbb{R}$ and of $\mathbf{n, l}\in \mathbb{Z}^{d}$ such that 
\begin{equation}
|\psi_{\mathbf l}(\lambda_{\mathbf l},\mathbf{n})|\leq Ce^{-c|\mathbf{n - l}|},\; \text{for any}\; \mathbf{n, l}\in 
\mathbb{Z}^{d} \label{eq:th-ulema}
\end{equation}
i.e., the operators has uniformly localised eigenfunctions \eqref{eq:eule}.

\item[(ii)] 
Let $H_\omega$ be the $1$-dimensional Maryland model \eqref{eq:maryland},    
$\gamma(\lambda, g) \ge \gamma(0, g) > 0$ be its Lyapunov exponent and $\varepsilon_0>0$ be the solution of the equation $\beta(\alpha)=\gamma(\varepsilon_0, g)$ with $\beta(\alpha)$ of \eqref{eq:beta}. Then the spectrum of $H_\omega$ is pure point on
$\sigma_{\pm}({H_\omega}) :=\{\pm \lambda \ge \varepsilon_0 > 0\}$, 
and if $\varepsilon_0 = 0$, then $\sigma_{pp}(H_\omega) = \sigma_-(H_\omega)\cup\sigma_+(H_\omega) = \R,$ see e.g.  \cite[Section 18.C]{Pa-Fi:92}. 

If $\{\lambda _{{l}}(\omega )\}_{{l}\in \mathbb{Z}}$ are its eigenvalues and  $\{\psi_{{l}}(\omega )\}_{{l}\in \mathbb{Z}}$, $\psi_{{l}}(\omega )=\{\psi _{{l}}(\omega , {n})\}_{{n}\in \mathbb{Z}}$ 
are the corresponding eigenfunctions of $H_\omega$, then there exist  $C<\infty $ and $c>0$ independent of $\lambda \in \mathbb{R}$ and of ${n, l}\in \mathbb{Z}$ such that 
\begin{equation}  \label{eq:ulema-1}
|\psi_l(\omega,m)| \le Ce^{-c |m-l|}, \; m \in \mathbb{Z}, \; |\lambda| \ge
\varepsilon_0,
\end{equation}
with $\omega$-independent parameters $C \le \infty, \; c >0$ and $
\varepsilon_0 \ge 0$, i.e., the operators has uniformly localised eigenfunctions \eqref{eq:eule}.
\end{itemize}
\end{theorem}

The description of eigenfunctions of the Maryland model $($see  \eqref{eq:psize} -- \eqref{eq:maref-1}$)$ provides the simplest and a quite
explicit illustration of the notion of the localisation centers and the
uniformly localised eigenfunctions. These spectral objects can be defined for 
any selfadjoint operator $H$ acting on $\ell^2(\mathbb{Z}^d)$. Namely, $H$ is said
to have \textit{uniformly localised eigenfunctions} $($ULE$)$ if and only if $H$ has
a complete set
\begin{equation}  \label{eq:psi-l-m}
\{\psi_l\}_l,\; \psi_l=\{\psi_l(\mathbf{m})\}_{\mathbf{m} \in \mathbb{Z}^d},
\end{equation}
of orthonormal eigenfunctions such that there exist $C < \infty, \; c>0 $
and $\mathbf{m}_l\in\mathbb{Z}^d, \; l=1,2,\dots$ 
providing the bound 
\begin{equation}  \label{eq:ule}
|\psi_l(\mathbf{m})| < C e^{-c|\mathbf{m }- \mathbf{m}_l|}, \; \mathbf{m}
\in \mathbb{Z}^d.
\end{equation}
Thus, the eigenfunctions are ``localised about" points $\{\mathbf{m}_l\}_l$
-- the ``localisation centers", which can provide a convenient indexation of
the eigenfunctions and the corresponding eigenvalues.

If $H_\omega$ is a self-adjoint ergodic operator, then we say that $H_\omega$
has the ULE if and only if it possesses with probability $1$ a complete set
of orthonormal eigenfunctions $\{\psi_l(\omega)\}_l,\, \psi_l(\omega) =
\{\psi_l(\omega, \mathbf{m})\}_{\mathbf{m}\in\Z^d}$ $($cf. \eqref{eq:psi-l-m}$)$ satisfying with the same probability the
bound 
\begin{equation}  \label{eq:eule}
|\psi_l(\omega, \mathbf{m})| < C e^{-c|\mathbf{m }- \mathbf{m}_l(\omega)|},
\; \mathbf{m}\in\mathbb{Z}^d,
\end{equation}
with $\omega$-independent $C < \infty$ and $c>0$ but possibly $\omega$%
-dependent localisation centers $\{\mathbf{m}_l\}_l$.

Note that the ULE is a quite special property. For instance, it is not the
case for the Schr\"{o}dinger operators with random potential $($the Anderson
model$)$ in any dimension and for the almost Mathieu operator \eqref{eq:amo},
see e.g. \cite{dRJLS}. On the other hand, the ULE is most explicit and
useful in the Maryland model and certain other models with almost periodic
potentials, see e.g., \cite{Da-Ga:11,Ha:16}.

\begin{rmk}
\label{rmk:ule2} In the context of this paper the ULE property, if it
provides one-to-one labeling of all eigenfunctions, leads directly and
simply to the exponential bounds \eqref{eq:exp-decay} and %
\eqref{eq:ef-cor-decay} for the Fermi projection and eigenfunction
correlator. It suffices just to apply the argument used in formulae %
\eqref{eq:pqexp} -- \eqref{eq:pqexp2} proving \eqref{eq:exp-decay} and %
\eqref{eq:ef-cor-decay} for the Maryland model. Hence, in this case,
Criterion~\ref{cr:1} or Criterion~\ref{cr:2} yield the Area Law and the
exponential dynamical localisation in expectation.
\end{rmk}

\noindent The proofs of Theorems~\ref{th:ap-d}-\ref{th:ule-d-maryland} are given in Section~\ref{sec:almost-periodic}.

}

{
\begin{rmk} Using recent results of Jitomirskaya and Kachkovskiy \cite[Corollary~5.14]{JKa2}, it is possible to obtain a stronger version of Theorem~\ref{th:ap-1d} $(ii)$.

\end{rmk}}

\smallskip In our next result, we consider one-dimensional discrete
Schr\"odinger operators $\eqref{eq:schrodinger1}$ with potentials generated
by a subshift of finite type (see Section~\ref{sec:subshift}), a very
interesting class of hyperbolic dynamical systems exhibiting the so-called
classical dynamical chaos, see e.g. \cite{Br-St:02}. For these ergodic
operators, the spectral localisation {
outside of a discrete set of energies and for a range of coupling constants} was recently proved in \cite{ADZ1, ADZ}%
. 

\smallskip

\begin{theorem}
\label{subshit-ee} Let $(\Omega_A, T)$ be a subshift of finite type and $%
\mathbf{P}$ be a $T$-ergodic measure that has the bounded distortion
property $\eqref{eq:bds}$. Assume that $T$ has a fixed point on $\Omega_A$
and let $v:\Omega_A \to \mathbb{R}$ be an 
 $\alpha$-H\"older continuous, $0 < \alpha \leq 1$, that is globally bunched, non-constant, or locally constant $($see Definition~$\ref{def:LC-and-SH})$. 
Consider the one-dimensional Schr\"odinger operator $\eqref{eq:schrodinger1}$ whose potential is
generated by a subshift of finite type, $V_{\omega}(n)=v(T^n\omega)$. 

Then there exists a bottom part $(\varepsilon_-, \varepsilon_0]$ of $I(H)$
of $\eqref{eq:kap}$ such that if the Fermi energy $\varepsilon_F\in(%
\varepsilon_-, \varepsilon_0]$, then the expectation of the entanglement
entropy $\eqref{eq:ee}$ of the free lattice fermions whose one-body
Hamiltonian is $H_\omega$ obeys the Area Law $\eqref{eq:al}$. 
\end{theorem}


To prove Theorem~\ref{subshit-ee} in Section~\ref{sec:subshift} we have
to substantially extend the existing spectral theory. We believe that this
extension (see Theorem~\ref{th:subshift-ee} below) is of independent
interest for the spectral theory of this class of operators and completes
the results of important recent works \cite{ADZ,ADZ1}.

\begin{theorem}
\label{th:subshift-ee} Let $(\Omega_A, T)$ be a subshift of finite type and $
\mathbf{P}$ be a $T$-ergodic measure that has the bounded distortion
property $\eqref{eq:bds}$. Assume that $T$ has a fixed point on $\Omega_A$
and let $v:\Omega_A \to \mathbb{R}$ be an 
$\alpha$-H\"older continuous, $0 < \alpha \leq 1$, that is globally bunched, non-constant, or locally constant  $($see Definition~$\ref{def:LC-and-SH})$. 
Consider the one-dimensional 
Schr\"odinger operator $H_\omega\ \eqref{eq:schrodinger1}$, whose potential
is generated by a subshift of finite type, $V_{\omega}(n) = v(T^n\omega)$,
and let $I\subset\mathbb{R}$ be a set where we have the uniformly positive
Lyapunov exponent and the uniform large deviation-type estimate $($see
Definition~$\ref{def:ple-uld})$. 
Then there exist $C<\infty, c > 0$ such that for any $m,n \in\mathbb{Z}_+$
we have 
\begin{equation}  \label{eq:cor-doubl-intr}
\mathbf{E}\big\{Q_I(m,n)\big\} \leq C e^{-c|m - n|},
\end{equation}
where $Q_I$ is the eigenfunction correlator $($defined below in $%
\eqref{eq:ef-cor})$ corresponding to $H_\omega$. In particular, the operator 
$H_\omega$ exhibits the exponential dynamical localisation in expectation $%
\eqref{eq:exp-dyn-loc} $ on $I$. 
\end{theorem}

{
\begin{rmk}
The existence of the compact set $I\supset\sigma(H_\omega)$ on which we
have uniformly positive Lyapunov exponent and the uniform large
deviation-type estimate is guaranteed by \cite{ADZ1, ADZ} $($see %
\eqref{eq:J-set} and below$)$, thereby providing the validity of Theorem~\ref{th:subshift-ee}. 
\noindent The examples of subshifts of finite type to keep in mind are: 

\begin{itemize}
\item The doubling map 
\begin{equation*}
T_2: \mathbb{T }\rightarrow \mathbb{T}, \quad T_2\omega = 2\omega (\mathrm{mod}\, 1),\,\,\, 
\end{equation*}

\item Arnold's cat map 
\begin{align*}
&\hspace{1cm}T_{\text{cat}}: \mathbb{T}^2 \rightarrow \mathbb{T}^2,\,\, T_{
\text{cat}}(\omega_1, \omega_2) 
= ((2\omega_1 + \omega_2) (\mathrm{mod}\,1), 
(\omega_1 + \omega_2)
(\mathrm{mod}\,1)),
\end{align*}
\end{itemize}
where the corresponding invariant measures are the Lebesgue measures on the
one- and two-dimensional tori respectively.

In these cases we define the potential by $V_\omega(n) = g v(T^n\omega)$
with $v:\Omega_A\rightarrow\mathbb{R}$ satisfying the assumptions of
Theorem~\ref{subshit-ee}. Then, there exists $g_0 > 0$ such that for all $0
< g \leq g_0$ Theorems~\ref{subshit-ee} and \ref{th:subshift-ee} hold true.
Moreover, as pointed out in \cite[Remark~2.18]{ADZ}, $g_0$ is not too small
if $\alpha$ is not too small.
\end{rmk}

An important ingredient of the proof of Theorem~\ref{th:subshift-ee} is the following lemma of independent interest. It provides a general condition for ``strong" localisation, more precisely, for exponential dynamical localisation in expectation \eqref{eq:exp-dyn-loc} for a bounded ergodic one-dimensional discrete Schr\"odinger operator n terms of the absence of quantum-mechanical tunneling between the distant domains of the system. Various similar conditions were often used in the analysis of Anderson localisation since the pioneering paper by Fr\"{o}hlich and Spencer \cite{Fr-Sp:83}. 

Let $M = \{M(j,k)\}_{j, k = -n}^n$ be a $(2n + 1)\times(2n + 1)$ matrix. 
Denote by 
\begin{equation}  \label{eq:bad}
\begin{split}
&\hspace{-0.5cm}\mathrm{Bad}\,(M, \epsilon) = \\
& \{\lambda\in\mathbb{R}:\, \max(|(M - \lambda)^{-1}(0, n)|, |(M -
\lambda)^{-1}(0, -n)|) \geq \epsilon\},
\end{split}%
\end{equation}
where $(M - \lambda)^{-1}(k, l)$ is the $(k, l)$ element of the matrix $(M-
\lambda)^{-1}$, the set of the ``badly"-controlled spectral parameters.

\begin{lemma}
\label{lem:bad-res} Let $H$ be a bounded ergodic one-dimensional discrete
Schr\"odinger operator $\eqref{eq:schrodinger1}$. Assume that for any $n
\geq n_0$ and for any $|k|, |l| \leq n^2$ such that $|k - l| \geq 10 n$ and
for some $c > 0$ 
\begin{equation*}  
\begin{split}
&\left\{\cap_{i, j\in\{0,1\}}\mathrm{Bad}\,(H^{[k - n + i, k + n - j]},
e^{-cn})\right\} \\
& \hspace{1cm} \cap \left\{ \cap_{i, j\in\{0,1\}}\mathrm{Bad}\,(H^{[l - n +
i, l + n -j]}, e^{-cn})\right\} = \emptyset,
\end{split}%
\end{equation*}
where $H^{[m, p]}$ denotes the restriction of $H$ to the interval $[m,p]\cap%
\mathbb{Z}$ with Dirichlet boundary conditions. Then the corresponding
eigenfunction correlator $Q_I$ of $\eqref{eq:ef-cor}$ decays exponentially,
namely for any $l \geq 100 n_0$ 
\begin{equation*}  
Q_I(0, l) \leq 16 e^{-\frac{c}{20}l}\equiv 16 e^{-c^{\prime }l}.
\end{equation*}
\end{lemma}

The proof of Lemma~\ref{lem:bad-res} is given in Section~\ref{sec:lem-pf}.

\begin{rmk}Note that for the case of bounded one-dimensional Schr\"odinger operator with independent identically distributed potential $($aka Anderson model$)$ Lemma~\ref{lem:bad-res} combined with large deviation principle $($see \cite{JF} an referenes therein$)$ yields a short  proof of the exponential dynamical localisation in expectation \eqref{eq:exp-dyn-loc}, the strongest form of Anderson localisation. The first proof of localisation for one-dimensional Anderson model is going back to \cite{GMP}, and \cite{KS}. Since then the problem has been attracting a considerable attention to this day, see e.g. \cite{JF} and references therein.
\end{rmk}

}

\subsection{Comments and Related Results}

\label{ssec:comm}

We will now comment on various issues related to the above results.

Formulae $\eqref{eq:ee}$ - $\eqref{eq:h}$ for the entanglement entropy of
free fermions make natural the following two-step strategy of the large-size
asymptotic analysis of the entropy.

The first step is a detailed spectral analysis of the corresponding one-body
Hamiltonian $\eqref{eq:H}$, in particular, obtaining sufficiently complete
information about the matrix (kernel) of the spectral (Fermi) projection $%
\eqref{eq:FP}$ or the eigenfunction correlator $\eqref{eq:ef-cor}$ providing
the bounds $\eqref{eq:cor-doubl-intr}$ or $\eqref{eq:exp-decay}$. Note that
this problem has essentially spectral content; moreover, it is one of the
main problems of the spectral theory of finite-difference and differential
operators.

The second step is an asymptotic analysis of the non-linear and non-smooth
functional of the Fermi projection $\eqref{eq:FP}$ given by the r.h.s. of $%
\eqref{eq:ee}$.

Note that these two problems have various amounts of difficulty for
different classes of operators and are not always explicitly seen in the
work in question. For instance, the Fermi projection of a translation
invariant operator $\eqref{eq:con}$ can be easily obtained by using the
Fourier transformation and for $d=1$ is just 
\begin{equation}
P(\varepsilon _{F})=\{P(m,n)\}_{m,n\in \mathbb{Z}},\;\;P(m,n)=\frac{\sin
k_{F}(m-n)}{\pi (m-n)},  \label{eq:FP1-con}
\end{equation}
where $k_{F}\in \mathbb{R}$ is the Fermi momentum determined by the Fermi
energy $\varepsilon _{F}$ and the symbol of the operator. 
However, this simple and explicit expression proved to be of little
use in solving the second problem, whose solution requires an involved
analysis \cite{La-Wi:80}, which becomes quite hard and includes a
considerable amount of microlocal analysis in the case of translational
invariant pseudodifferential operators in $\ell^2(\mathbb{Z}^d)$ and $L^2(%
\mathbb{Z}^d)$, see \cite{Le-Co:14} -- \cite{Le-Co:17}, and \cite{So:17}.

On the other hand, for non-trivial ergodic operators, including the Schr\"{o}%
dinger operators $H_\omega$ with ergodic random potentials, whose spectrum
have a pure point component with exponentially decaying eigenfunctions
(Anderson localisation), one expects that if the Fermi energy belongs to the
pure point component $\sigma _{pp}(H_{\omega})$ of the $\omega$-independent
spectrum of $H_\omega$, then Fermi projection $\eqref{eq:FP}$, i.e., 
\begin{equation}
P_{\omega}(\varepsilon _{F})=\chi _{(\varepsilon _{-},\varepsilon
_{F}]}(H)=\{P(\mathbf{m,n})\}_{\mathbf{m,n}\in \mathbb{Z}^{d}}
\label{eq:FP-ent}
\end{equation}%
admits the bound 
\begin{equation}
\mathbf{E}\{|P(\mathbf{m,n})|\}\leq Ce^{-c|\mathbf{m}-\mathbf{n}|},\;\mathbf{%
m,n}\in \mathbb{Z}^{d},  \label{eq:exp-decay}
\end{equation}%
where $C<\infty ,\;c>0$, and do not depend on $\mathbf{m,n}\in \mathbb{Z}%
^{d} $, but may depend on $\varepsilon _{F}$.

This motivates

\begin{definition}
\label{def:points} We say that $\lambda _{0} \in \mathbb{R}$ is

\smallskip $(i)$ a FPED $($Fermi projection exponential decay$)$ point, if
the Fermi projection $P_{\omega}(\varepsilon _{F})|_{\varepsilon
_{F}=\lambda _{0}}$ of \eqref{eq:FP-ent} satisfies \eqref{eq:exp-decay},

\smallskip $(ii)$ an AL $($Area Law$)$ point if the expectation $\mathbf{E}%
\{S_{\Lambda }(\varepsilon _{F})|_{\varepsilon _{F}=\lambda _{0}}\}$ of the
entanglement entropy \eqref{eq:ee} -- \eqref{eq:h} obeys the Area Law %
\eqref{eq:al}.
\end{definition}

\noindent Here are two types of FPED-points.

\smallskip 1) Any point of an internal gap of the spectrum $\sigma(H)$ of
the one-body Hamiltonian $H$ of \eqref{eq:lro}-\eqref{eq:W}. Indeed, it
follows from the Combes-Thomas estimate \cite[Section~10.3]{AW} and \cite[%
Appendix B]{Fr-Sp:83}, that the Green function $(H-z)^{-1}(\mathbf{m}, 
\mathbf{n}),\, \mathbf{m}, \mathbf{n}\in\mathbb{Z}^d$, of a (non necessarily
ergodic) operator $\eqref{eq:lro}$ decays exponentially if $z\in\mathbb{C}%
\setminus\sigma(H)$. 
Combining this with the contour integration trick (see e.g, \cite[Section 13]%
{AW}) we obtain the exponential decay of the Fermi projection $%
P(\varepsilon_F)$ if the Fermi energy $\varepsilon_F$ falls in a gap.

\smallskip 2) Any point of the pure point spectrum with the exponentially
decaying eigenfunctions of $H_{\omega}$ with independent identically
distributed (monotone) random potentials, see \eqref{eq:lro} -- \eqref{eq:W}
and \eqref{eq:ergodic-potential-1}, i.e., for the case of so-called
``strong" Anderson localisation, mentioned above. 
%
For this case, the proof of \eqref{eq:exp-decay} is among the top results of
the field, see e.g. \cite[Section 13]{AW} and \cite{St:11} for any $d \ge 1$%
, and \cite[ Section 12]{AW} for $d=1$, where the whole spectrum is strongly
localized.

\smallskip

Note that the same (FPED) property holds for a more general
classes of random operators, the so-called Wegner orbital model and
Schr\"odinger operators with non-monotone independent identically
distributed random potentials, see \cite{wegner-orbital, ESS}. Thus,
Criterion~\ref{cr:1} below yields the validity of the Area Law for the
expectation of the entanglement entropy of these models.

\smallskip

\smallskip The bound \eqref{eq:exp-decay} provides the solution of the first
problem. Moreover, it proves to be the main ingredient in the solution of
the second problem, hence, in obtaining the Area Law for the Schr\"odinger
operators with random potentials, see \cite{EPS}.

Denote $\sigma_{FPDE}(H)$ and $\sigma_{AL}(H)$ 
the sets of the FPDE- and AL-points of our one-body Hamiltonian. Then we can
formulate the main result of  \cite{EPS} (the solution of the second problem for the random
Schr\"odinger operator) as the proof of  the inclusion 
\begin{equation*}  
\sigma_{FPED}(H)\subseteq \sigma_{AL}(H).
\end{equation*}
However, as the analysis of the proof of the Area Law in \cite[Results 2--3]%
{EPS} shows, it is applicable to any ergodic finite-difference operator for
which the estimate $\eqref{eq:exp-decay}$ is valid.

The goal of this work is to use the same approach, i.e., the analysis of $%
\sigma_{FPED}$, to prove the Area Law for certain dynamically generated
(including some quasi-periodic) potentials in \eqref{eq:lro} -- \eqref{eq:W}.

\medskip
It is convenient to formulate the above as

\begin{criterion}
\label{cr:1} Let $S_{\Lambda }(\varepsilon _{F})$ be the entanglement
entropy \eqref{eq:ee} -- \eqref{eq:h} of free lattice fermions whose
one-body Hamiltonian $H_{\omega }$ is a self-adjoint ergodic operator $%
\eqref{eq:herg}$ in $\ell ^{2}(\mathbb{Z}^{d})$, and let $P_{\omega }$ be
its Fermi projection $\eqref{eq:FP}$. Assume that the Fermi energy $%
\varepsilon _{F}$ belongs to $\sigma_{FPED}$(H), i.e., 
\eqref{eq:exp-decay} holds.

Then, for $\Lambda = \Lambda_M = [-M, M]^d\subset\mathbb{Z}^d,\, L = 2M + 1$%
, the expectation of the corresponding entanglement entropy $%
S_\Lambda(\varepsilon_F)$ of \eqref{eq:ee} obeys the Area Law: 
\begin{equation*} 
0 < \lim_{L\to\infty}L^{-(d - 1)}\mathbf{E}\,\{S_{\Lambda_M}(\varepsilon_F)%
\} := s_d(\varepsilon_F) < \infty,
\end{equation*}
see \cite{EPS} for an explicit formula for $s_d(\varepsilon_F)$.
\end{criterion}

Our strategy is essentially determined by Criterion~\ref{cr:1}, i.e., it
consists in proving the exponential bound $\eqref{eq:exp-decay}$ for several
classes of non-random ergodic operators. 

It turns out that in a number of cases under consideration it is convenient
to use a spectral characteristic somewhat more general but simpler than the
Fermi projection and known as the\textit{\ eigenfunction correlator}, introduced by Aizenman \cite{A94}, see also
\cite[Section~1.4, and Chapter~7]{AW}, and Definition $\eqref{def:ef-cor}$
below. The characteristic proves to be important in a number of problems of
spectral analysis of ergodic operators and its applications, see e.g. \cite%
{AW,Pa-Fi:92,St:11} and references therein.

\begin{defin}
\label{def:ef-cor} Let $H$ be a self-adjoint operator on $\ell^2(\mathbb{Z}%
^d)$. Then the \textrm{\ eigenfunction correlator} for a Borel set $I\subset%
\mathbb{R}$ and any $\mathbf{m}, \mathbf{n}\in\mathbb{Z}^d$ is 
\begin{equation}  \label{eq:ef-cor}
Q_I(\mathbf{m}, \mathbf{n}) := \sup_\phi\{|\langle \delta_\mathbf{m},
\phi(H)\chi_I (H) \delta_\mathbf{n}\rangle| :\, \phi\in C(\mathbb{R}), \; \;
\|\phi\|_\infty \leq 1\},
\end{equation}
where $C(\mathbb{R})$ is the space of continuous real-valued functions, $%
\{\delta_{\mathbf{n}}\}_{\mathbf{n}\in\mathbb{Z}^d}$ is the canonical basis
of $\ell^2(\mathbb{Z}^d)$, and $\chi_I$ is the indicator of $I$.
\end{defin}

If the spectrum of $H$ in $I$ is pure point, 
and $\{\lambda_l\}_l$ and $\{\psi_l\}_l$ are the eigenvalues (counted with
their multiplicity) and normalized eigenfunctions respectively, then 
\begin{equation}  \label{eq:ef-cor1}
Q_I(\mathbf{m}, \mathbf{n})=\sum_{\lambda_l \in I} |\psi_l (\mathbf{m})| \,|
\psi_l (\mathbf{n})|.
\end{equation}
This explains the name of $Q_I$.

Next, we have from spectral theorem, $\eqref{eq:FP-ent}$, and $%
\eqref{eq:ef-cor1}$ with $I=[\varepsilon_-, \varepsilon_F)$ 

\begin{equation*} 
|P(\mathbf{m,n})|\le Q_{[\epsilon_-,\varepsilon_F)}(\mathbf{m,n}).
\end{equation*}
We also have 
\begin{equation*} 
\sup_{t \in \mathbb{R} }\{|\langle \delta_\mathbf{m}, e^{-it H}\chi_I (H)
\delta_\mathbf{n}\rangle|\le Q_{I}(\mathbf{m,n}),
\end{equation*}
where the l.h.s. is another spectral characteristics, important in
applications, in particular, in determining the \textit{exponential
dynamical localisation in expectation} (see $\eqref{eq:exp-dyn-loc}$),
one of the strongest forms of the Anderson localisation.

The above bounds and Criterion~\ref{cr:1} imply the following workable
criterion for the validity of the Area Law and the exponential dynamical
localisation in expectation.%

\begin{criterion}
\label{cr:2} Let $H_\omega$ be a self-adjoint ergodic operator $%
\eqref{eq:herg}$ in $\ell^2 (\mathbb{Z}^d)$ and $Q_{I}(\mathbf{m,n})$ be its
eigenfunction correlator $\eqref{eq:ef-cor}$ with a Borel set $I \in \mathbb{%
R}$. Then the bound 
\begin{equation}  \label{eq:ef-cor-decay}
\mathbf{E}\{|Q_{I}(\mathbf{m,n})|\} \le C e^{-c|\mathbf{m}-\mathbf{n}|}, \; 
\mathbf{m,n} \in \mathbb{Z}^d,\;
\end{equation}
where $C<\infty$ and $c>0$ do not depend on $\mathbf{m,n}$ $($but may depend
on $I)$, implies:

\smallskip

 $(i)$ Exponential dynamical localisation in expectation of $%
H_\omega$ on $I$: 
\begin{equation}  \label{eq:exp-dyn-loc}
\mathbf{E}\Big\{\sup_{t\in\mathbb{R}}|\langle \delta_\mathbf{m},
e^{-itH_\omega}\chi_I (H_\omega)\delta_\mathbf{n}|\Big\} \leq C e^{-c|%
\mathbf{m}-\mathbf{n} |}.
\end{equation}

$(ii)$ Area Law $\eqref{eq:al}$ for the expectation $\mathbf{E}%
\{S_\Lambda(\varepsilon_F)\}$ of the entanglement entropy $\eqref{eq:ee}$ -- %
\eqref{eq:h} of free lattice fermions having $H_\omega$ as their one-body
Hamiltonian if $I=(\varepsilon_-, \varepsilon_0]$ is the bottom of the
spectrum of $H_\omega$ and the Fermi energy $\varepsilon_F < \varepsilon_0$.
\end{criterion}
We will consider first the finite-difference operators \eqref{eq:lro} 
with quasi-periodic potentials for which exponential bounds for the Fermi
projection \eqref{eq:FP-ent} and/or the eigenfunction correlator $%
\eqref{eq:ef-cor}$ are essentially known, although sometimes in the form
that has to be made suitable for our purposes. These results are presented
largely for the sake of completeness, and also to demonstrate yet another
application of the spectral theory of ergodic operators, this time to the
quantum information science.

To explain our main new result, note first that quasi-periodic functions
follow the orbits of the irrational winding (total shifts) of the torus $%
\mathbb{T}^d$, a quite simple dynamical system. A considerably more complex
and rich class of dynamical systems with the same phase space $\mathbb{T}^d$
consists of the so-called subshifts of finite type \cite{Adler,
Br-St:02,Da:07}. The spectral analysis of these operators has been recently
developed in \cite{ADZ,ADZ1}. A number of important and highly non-trivial
facts, including the existence of Anderson (spectral) localisation was
established in these papers. However, the crucial for our purposes
exponential decay $\eqref{eq:ef-cor-decay}$ of the eigenfunction correlator
and/or Fermi projection $\eqref{eq:exp-decay}$ is lacking.

We will prove this property, hence by Criterion~\ref{cr:2} obtain the Area
Law, as well as the exponential dynamical localisation in expectation,
which we believe is of independent interest.

{\ The exponential bounds \eqref{eq:exp-decay}, and \eqref{eq:exp-dyn-loc}
are well known in the spectral theory of random ergodic operators, first of
all the Schr\"odinger operators whose potential is a collection of
independent identically distributed random variables \cite{AW,ESS, St:11}. There
are several proofs of the bounds in this case, a quite streamline and
efficient one is based on the analysis of the fractional moments of the
resolvent of the operator. }

The proofs of $\eqref{eq:exp-decay}$ and $\eqref{eq:exp-dyn-loc}$ for the
quasi-periodic operators are based on the positivity of the Lyapunov
exponent (see, e.g. $\eqref{eq:lyap-exp}$) of the corresponding
finite-difference equation and on the related exponential decay of the
eigenfunctions of $H_\omega$. We will use this approach for Theorems~\ref%
{th:ap-d} and \ref{th:ap-1d} in Section~\ref{sec:almost-periodic}. 

The proof of Theorems~\ref{subshit-ee} and \ref{th:subshift-ee} for the
Schr\"odinger operators whose potential is generated by a subshift of finite
type is using as an input the uniform positivity of the Lyapunov exponent
and the uniform large deviation-type estimate from \cite{ADZ1, ADZ}. In addition, we
use special techniques of the theory of dynamical systems involving the
Markov partitions and related Markov chains \cite{Adler, Br-St:02}.
%
%
\section{Finite-difference operators with quasi-periodic potential}

\label{sec:almost-periodic}

In this section, we present the proofs of Theorems~\ref{th:ap-d} and \ref{th:ap-1d} on the Area Law for several classes of non-random ergodic
operators. We begin with the proof of Theorem~\ref{th:ule-d-maryland}, which is important as itself  and crucial in the proofs of Theorems~\ref{th:ap-d} and \ref{th:ap-1d}.  

To the best of our knowledge, the only class of ergodic operators
for which 
the Area Law is rigorously established (using Criterion~\ref{cr:1}) is that
consisting of the Schr\"odinger operators with independent identically
distributed random potentials (aka Anderson model) \cite{EPS}.

Here we show that the class of such operators is considerably larger. In
particular, it includes several families of quasi-periodic (and limit-periodic), i.e.,
deterministic, rather than random, operators. These are operators given by $
\eqref{eq:schrodinger}$, whose ergodic potential $
\eqref{eq:ergodic-potential-1} $ is generated by an irrational shift 
\begin{equation*} 
\, T(\omega) = \omega + \alpha \; \text{mod} \; 2 \pi, \; \omega\in\mathbb{T}
, \alpha\in\mathbb{R}\setminus\mathbb{Q},\, n\in\mathbb{Z},
\end{equation*}
that is, $\Omega = \mathbb{T}$ equipped with the Borel $\sigma$-algebra and $
\mathbf{P}$ being the normalised Lebesgue measure.

The two archetypal examples of such one-dimensional Schr\"odinger operators
are the Maryland model, whose potential is $\eqref{eq:maryland}$ and the
almost Mathieu operator with potential $\eqref{eq:amo}$.

{
We start with the proof of Theorem~\ref{th:ule-d-maryland} which we consider to be of independent interest.

\subsection{Uniformly localised eigenfunctions for the $d$-dimensional, $d\geq 1$, Maryland model.}

\begin{proof}[Proof of Theorem~\ref{th:ule-d-maryland} $(i)$]

\noindent According to \cite[Sections 18.A-18.B]{Pa-Fi:92}, the spectrum of the $d$-dimensional Maryland model 
occupies the whole $\mathbb{R}$, is pure point,
of multiplicity $1$, and is described as follows. There exist a 1-periodic and
monotone on the period function $\lambda _{0}:\mathbb{T\rightarrow R}$ \ and
\begin{equation}
\psi _{0}:\mathbb{R}\times \mathbb{Z}^{d}\rightarrow \mathbb{C},\;\sum_{
\mathbf{n}\in \mathbb{Z}^{d}}|\psi _{0}(\lambda ,\mathbf{n}
)|^{2}=1,\;\forall \lambda \in \mathbb{R},  \label{eq:psize}
\end{equation}
such that the eigenvalues $\{\lambda _{\mathbf{l}}(\omega )\}_{\mathbf{l}\in 
\mathbb{Z}^{d}}$ are 
\begin{equation}
\lambda _{\mathbf{l}}(\omega )=\lambda _{0}(\omega +\langle \mathbf{\alpha ,}
{\mathbf{l\rangle }}),\;\lambda _{{\mathbf{l}}_{1}}(\omega )\neq \lambda _{{
\mathbf{l}}_{2}}(\omega )\Longleftrightarrow {\mathbf{l}}_{1}\neq {\mathbf{l}
}_{2},  \label{eq:lal-1}
\end{equation}
and the corresponding eigenfunctions $\{\psi _{\mathbf{l}}(\omega )\}_{
\mathbf{l}\in \mathbb{Z}^{d}}$, $\psi _{\mathbf{l}}(\omega )=\{\psi _{
\mathbf{l}}(\omega ,\mathbf{n})\}_{\mathbf{n}\in \mathbb{Z}^{d}}$ are
\begin{equation}
\psi _{{\mathbf{l}}}(\omega ,{\mathbf{n}})=\psi _{0}(\lambda _{{\mathbf{l}}
}(\omega ),{\mathbf{n}}-{\mathbf{l}}).  \label{eq:maref-1}
\end{equation}

Note first that the function $\lambda _{0}$
of \eqref{eq:lal-1} is the functional inverse of the Integrated Density of
States $N(\lambda )$ of $H_\omega$ (see \cite[Sections 4.B -- 4.C]{Pa-Fi:92} and 
\cite[Sections 3.3 -- 3.4]{AW} for its definition and properties). We have for the Maryland model 
\begin{equation}
N(\lambda )=\int_{-\infty }^{\lambda }n(\lambda)d\lambda ,  \label{eq:ncm}
\end{equation}%
i.e., 
\begin{equation}
\lambda _{0}\circ N=1,  \label{eq:laN}
\end{equation}%
and
\begin{equation}
n(\lambda )=\frac{g}{(2\pi i)^{d}\pi }\int_{\mathbb{T}^{d}}|w(\mathbf{\eta }%
)-\lambda +ig|^{-2}\prod_{j=1}^{d}\frac{d\eta _{j}}{\eta _{j}},
\label{eq:dos}
\end{equation}%
is the Density of States of $H_\omega$ and $w:\mathbb{T}^{d}\rightarrow \mathbb{R}$
is given by the $d$-dimensional Fourier series with coefficients of %
\eqref{eq:W} 
\begin{align}
& \,w(\mathbf{\eta })=\sum_{\mathbf{n}\in \mathbb{Z}^{d}}W(\mathbf{n})\,\eta
_{1}^{n_{1}}\cdots \eta _{d}^{n_{d}},  \notag \\
\mathbf{\eta }& =(\eta _{1},\dots ,\eta _{d})\in \mathbb{T}^{d},\,\,\mathbf{n%
}=(n_{1},\dots ,n_{d})\in \mathbb{Z}^{d}.  \label{eq:lw}
\end{align}%
Next, the function $\psi _{0}$ of \eqref{eq:psize} is (see \cite[Section 18.B]%
{Pa-Fi:92}) 
\begin{equation}
\psi _{0}(\lambda ,{\mathbf{n}})=\frac{1}{(2\pi i)^{d}(\pi n(\lambda
)/g)^{1/2}}\int_{\mathbb{T}^{d}}\frac{e^{t(\lambda ,\mathbf{\eta })}}{(w(%
\mathbf{\eta })-\lambda -ig)}\prod_{j=1}^{d}\eta ^{n_{j}-1}d\eta _{j},
\label{eq:psiz}
\end{equation}%
where 
\begin{equation}
t(\lambda ,\mathbf{\eta })=\sum_{\mathbf{m}\in \mathbb{Z}^{d}} t_{\mathbf{m}%
}(\lambda )\prod_{j=1}^{d}\eta _{j}^{m_{j}},\; t_{\mathbf{m}} (\lambda)= 
\frac{l_{\mathbf{m}}(\lambda )}{(1-e^{ i\langle \mathbf{\alpha ,m}\rangle })}
\label{eq:t}
\end{equation}%
with 
\begin{equation}
l_{\mathbf{m}}(\lambda )=(2\pi i)^{-d}\int_{\mathbb{T}^{d}}\log c(\lambda ,%
\mathbf{\eta })\prod_{j=1}^{d}\eta ^{m_{j}-1}d\eta _{j},  \label{eq:dc}
\end{equation}%
and 
\begin{equation}
c(\lambda ,\mathbf{\eta })=-\frac{w(\mathbf{\eta })-\lambda +ig}{w(\mathbf{%
\eta } )-\lambda -ig}.  \label{eq:c}
\end{equation}%
Thus, the $c(\lambda ,\mathbf{\eta })$ of \eqref{eq:c} is the input of the
above analytic procedure given by formulae from \eqref{eq:dc} to %
\eqref{eq:psiz} that leads to the eigenfunctions \eqref{eq:maref-1}. 

Our goal is to prove the bound \eqref{eq:th-ulema}. This combined with
Remark~\ref{rmk:ule2} and Criterion~\ref{cr:2} will allow us to
establish the validity of the Area Law for the Maryland model.

We are going to use the following simple fact.

\begin{lemma}
\label{lemma:four} Let $f:\mathbb{T}^{d}\rightarrow \mathbb{C}$ be a
periodic function and $\{F_{\mathbf{n}}\}_{\mathbf{n}\in \mathbb{Z}^{d}}$ be
its Fourier coefficients. We have

\smallskip $($i$)$ If $|F_{\mathbf{n}}|\leq Fe^{-\rho _{1}|\mathbf{n}%
|},\;F<\infty ,\;\rho _{1}>0,\;\mathbf{n}\in \mathbb{Z}^{d},$ then $f$
admits the analytic continuation into 
\begin{equation*}
\mathbb{T}_{\rho _{1}}^{d}=\prod\limits_{j=1}^{d}\mathbb{T}_{\rho _{1}},\;%
\mathbb{T}_{\rho _{1}}=\{z\in \mathbb{C}:e^{-\rho _{1}}\leq |z|\leq e^{\rho
_{1}}\}, 
\end{equation*}
where for any $0 < \rho_2 < \rho_1$ and $a = \rho_1 - \rho_2 > 0$ it has the
bound
\begin{equation}
|f(\mathbf{z})|\leq F\left(\frac{e^{a}+1}{e^{a}-1}\right)^{d},\; \text{for any} \quad \mathbf{z}%
\in \mathbb{T}_{\rho _{2}}^{d}.  \label{eq:oc1}
\end{equation}

$($ii$)$ If $f$ admits the analytic continuation to $\mathbb{T}_{\rho
_{1}}^{d},\;\rho _{1}>0,$ then%
\begin{eqnarray}
|F_{\mathbf{n}}| &\leq &Fe^{-\rho _{1}|\mathbf{n}|},\;F<\infty ,\;\mathbf{n}%
\in \mathbb{Z}^{d},  \notag \\
F &=&\max_{|r|\leq \rho _{1},\;|\mathbf{\eta }|=1}|f(e^{r}\mathbf{\eta })|.
\label{eq:oc2}
\end{eqnarray}
\end{lemma}

\begin{proof}
Assertion (i) follows from a direct calculation, while assertion (ii) is
given by an appropriate deformation of the ``contour" $\mathbb{T}^{d}$ to
the poly-annulus $\mathbb{T}_{\rho _{2}}^{d}$, for any $\rho_2 < \rho_1$, in
the formula for the Fourier coefficients $F_{\mathbf{n}}$ of $f$.
\end{proof}

Now, applying Lemma~\ref{lemma:four} (i) to $w$ of \eqref{eq:lw} whose
Fourier coefficients satisfy \eqref{eq:W}, we find that $w$ admits the
analytic continuation to $\mathbb{T}_{\rho _{1}}^{d},\;0<\rho _{1}<\rho $,
where $\rho > 0$ is given in \eqref{eq:W}. Moreover, since the ``initial" $w$
of \eqref{eq:lw} is real-valued ($w:\mathbb{T}^{d}\rightarrow \mathbb{R}$),
we have for its analytic continuation%
\begin{equation}
\lim_{r\rightarrow 0}\mathrm{Im} \ w(e^{r}\mathbf{\eta })=0,\;|\mathbf{\eta }%
|=1.  \label{eq:imw}
\end{equation}%
We conclude that there exists $0<\rho _{2}<\rho _{1}$ such that for any $|r|
\leq \rho_2$ 
\begin{equation}
|w(e^{r}\mathbf{\eta })-\lambda \pm i g| \ge |g\pm \mathrm{Im} \, w(e^{r}%
\mathbf{\eta })|\ge |g|-|\mathrm{Im} \, w(e^{r}\mathbf{\eta })| \ge g/2,
\label{eq:2g}
\end{equation}%
where the last inequality follows from $\eqref{eq:imw}$. We have then from %
\eqref{eq:c} that $|\log c(\lambda ,\mathbf{\eta })|\leq |\log |c(\lambda ,%
\mathbf{\eta })|\ |+2\pi ,$ and for any $|r| \leq \rho_2$ 
\begin{eqnarray*}
|\log |c(\lambda ,e^{r}\mathbf{\eta })| | &=&\big|\log |1+2ig(w(e^{r}\mathbf{%
\eta })-\lambda -ig)^{-1}|\big| \\
&\leq &2|g|\ |w(e^{r}\mathbf{\eta })-\lambda -ig|^{-1}\leq 4,
\end{eqnarray*}%
where in obtaining the last bound of the r.h.s. above we used \eqref{eq:2g}.
Combining the two last bounds, we get for any $|r| \leq \rho_2$ 
\begin{equation*}
|\log c(\lambda ,e^{r}\mathbf{\eta })|\leq C,
\end{equation*}%
where $C$ is a constant. Now combining \eqref{eq:dc} and Lemma~\ref{lemma:four}
(ii) we obtain 
\begin{equation*}  
|l_{\mathbf{m}}(\lambda )|\leq L e^{-\rho _{2}|\mathbf{m}|},
\end{equation*}
where $L < \infty$ and $\rho _{2} > 0$ do not depend on $\lambda $ and $\mathbf{m}$.

\noindent
This and the Diophantine condition \eqref{eq:Dob} implies for the Fourier
coefficients $t_{\mathbf{m}}(\lambda) $ of  \eqref{eq:t} 
\begin{equation}  \label{eq:tmd1}
|t_{\mathbf{m}}(\lambda )|\leq Le^{-\rho _{2}^{\prime }|\mathbf{m}|},\; 0<\rho^{\prime }_2 <\rho_2.
\end{equation}

Then, Lemma \ref{lemma:four} (ii) imply that $t(\lambda ,\cdot )$ of %
\eqref{eq:t} admits the analytic continuation into $\mathbb{T}_{\rho
_{3}}^{d}$ with some $0<\rho_{3}<\rho^{\prime}_{2}$ and is bounded there
for any $|r|\leq\rho_3$, i.e., 
\begin{equation}
|t(\lambda ,e^{r}\mathbf{\eta})|\leq T<\infty, \; |r|\leq\rho_3, \label{eq:oct}
\end{equation}%
where $T$ and $\rho_3$ are $\lambda $-independent according to \eqref{eq:oc1}.

Next, \eqref{eq:2g}, and the analyticity of $w$ in $\mathbb{T}_{\rho_{1}}^{d}$ imply that 
\begin{equation*}
h(\lambda ,\mathbf{z})=\frac{e^{t(\lambda ,\mathbf{z})}}{(\pi n(\lambda
)/g)^{1/2}(w(\mathbf{z})-\lambda -ig)}  
\end{equation*}
is analytic in $\mathbf{z}$ in $\mathbb{T}_{\rho _{3}}^{d}$ for any $\lambda
\in \mathbb{R}$. Since the r.h.s. of \eqref{eq:psiz} is the Fourier
coefficient of $h(\lambda ,\mathbf{\cdot })$, \eqref{eq:oc2} yields the
following bound for the r.h.s. of \eqref{eq:psiz}%
\begin{equation}
\Psi e^{-\rho _{3}|\mathbf{n}|},\;\mathbf{n}\in \mathbb{Z}^{d},\;\Psi
=\max_{\lambda \in \mathbb{R},\;|r|\leq \rho _{3},\;|\mathbf{\eta }%
|=1}|h(\lambda ,e^{r}\mathbf{\eta })|.  \label{eq:ocf}
\end{equation}%
It follows from \eqref{eq:dos} that

\begin{enumerate}
\item $(n(\lambda ))^{-1/2}\geq (\pi g)^{1/2},$ and $(n(\lambda
))^{1/2}=|\lambda |^{-1}(1+o(1)),\;|\lambda |\rightarrow \infty $, by %
\eqref{eq:dos};

\smallskip

\item \ $|w(\mathbf{z})-\lambda -ig|,\;|\mathbf z|\in \lbrack e^{-\rho _{3}},e^{\rho
_{3}}]$ is bounded from below by \eqref{eq:2g} and is $|\lambda
|(1+o(1)),\;|\lambda |\rightarrow \infty $;

\smallskip

\item \ $e^{t(\lambda ,\mathbf{z})}\leq e^{T}$ by \eqref{eq:oct}.
\end{enumerate}

\noindent This and \eqref{eq:ocf} imply that $\Psi $ is independent of $%
\lambda \in \mathbb{R}$ and $\mathbf{n}\in \mathbb{Z}^{d}$.


Thus we obtain
\begin{equation}
|\psi _{0}(\lambda ,\mathbf{n})|\leq Ce^{-c|\mathbf{n}|},\;\mathbf{n}\in 
\mathbb{Z}^{d},  \label{eq:ulema-1}
\end{equation}
 with $\Psi $ of \eqref{eq:ocf} as $C$ and $\rho
_{3}$ of \eqref{eq:oct} as $c$.
Combining \eqref{eq:maref-1} and 
\eqref{eq:ulema-1}, we obtain 
\begin{equation}
|\psi _{{\mathbf{l}}}(\omega ,{\mathbf{n}})|\leq Ce^{-c|\mathbf{n-l}|},
\label{eq:uule}
\end{equation}%
where $C<\infty $ and $c>0$ are independent of $\omega $. 
\end{proof}


\begin{proof}[Proof of Theorem~\ref{th:ule-d-maryland} $(ii)$]
According to  \cite[Section 18.C]{Pa-Fi:92}, the pure point spectrum consists of two semi-infinite components \eqref{eq:gabe}

Since in this case operator $W$ in \eqref{eq:lro} is the one-dimensional
discrete Laplacian, we set $W_0=0, \; W_{\pm 1}=1$ in \eqref{eq:W}. Hence, $w$ of \eqref{eq:lw} is 
\begin{equation}  \label{eq:w1}
w(\eta)=\eta+\eta^{-1}.
\end{equation}
Thus, $w$ is analytic everywhere except zero and infinity, i.e., in this
case the analog of $\rho_1$ of the proof above is infinity.


Next, the integral in the r.h.s. of formula \eqref{eq:dc} for $d=1$ and $w$
of \eqref{eq:w1} can be calculated yielding (cf. \eqref{eq:dc}) 
\begin{equation*}
l_m=\frac{2i}{|m|}e ^{-\gamma|m|}\sin|m| \varphi.
\end{equation*}
Here $e^{-\gamma -i\varphi}=:\eta_0$ is the root of equation $\eta
+\eta^{-1}=\lambda+ig$, such that $|\eta_0|<1$ and $\gamma:=\gamma(
\lambda,g) $ is the Lyapunov exponent of the corresponding finite-difference
equation 
\begin{equation}  \label{eq:maryeq}
u_{n+1}+u_{n-1} + g \tan(\alpha n + \omega)u_n=\lambda u_n.
\end{equation}
We have 
\begin{equation}\label{eq:le-mary}
\mathrm{sinh} \gamma(\lambda,g)=s>0,
\end{equation}
where $s > 0$ is such that 
\begin{equation}\label{eq:le-mary2}
s^4+(4-\lambda^2 -g^2)s^2 -g^2=0,
\end{equation}
implying that $\gamma(\lambda,g)$ is an even and convex function of $\lambda
\in \mathbb{R}$ and 
\begin{equation}\label{eq:le-mary3}
\gamma(\lambda,g) \ge \gamma(0,g)=\mathrm{arc \, sinh} (|g|/2)>0,\; g \neq 0.
\end{equation}
Denote $\lambda_0 \ge 0$ the root of equation $\gamma(\lambda,g)=\beta(%
\alpha)$, with $\beta(\alpha)$ defined in \eqref{eq:beta} and fix $%
\varepsilon_0 > \lambda_0$. Then, setting 
\begin{equation*}
\rho^{\prime }:=\gamma(\varepsilon_0,g)-\beta(\alpha)>0,
\end{equation*}
we obtain from \eqref{eq:t} and \eqref{eq:dc} for the Fourier coefficients $
t_m(\lambda)$ of function $t(\mathbf{\eta},\cdot)$ of \eqref{eq:t} for $d=1$
\begin{equation*}
|t_m(\lambda)| \le Ce^{-\rho^{\prime }|m|}, \; |\lambda| \ge \varepsilon_0.
\end{equation*}
This bound is an analog of \eqref{eq:tmd1}. Hence, repeating the argument of
the previous proof  starting from \eqref{eq:t} we get
\begin{equation} \label{eq:ulema2}
|\psi_l(\omega,m)| \le Ce^{-c |m-l|}, \; m \in \mathbb{Z}, \; |\lambda| \ge
\varepsilon_0,
\end{equation}
with $\omega$-independent parameters $C < \infty, \; c >0$ and $
\varepsilon_0 \ge 0$.
\end{proof}
} 

\subsection{Proof of Theorem~\ref{th:ap-d}} 


\medskip

\begin{proof}[Proof of Theorem~\protect\ref{th:ap-d}~$(i)$]
Spectral
theorem for the Fermi projection
\begin{equation}
P(\omega ,\mathbf{m,n})=\sum_{\lambda _{\mathbf l}(\omega )\leq \varepsilon
_{F}}\psi _{{\mathbf{l}}}(\omega ,\mathbf{m})\overline{\psi _{{\mathbf{l}}%
}(\omega ,{\mathbf{n}})}  \label{eq:pf_spec}
\end{equation}
combined with \eqref{eq:uule} imply 
\begin{align}
& \left\vert P(\omega ,\mathbf{m,n})\right\vert \leq Q_{[\varepsilon
_{-},\varepsilon _{F})}(\omega ,\mathbf{m,n})=\sum_{\mathbf{l}:\lambda
_{\mathbf l}(\omega )\leq \varepsilon _{F}}|\psi _{{\mathbf{l}}}(\omega ,\mathbf{m}%
)\psi _{{\mathbf{l}}}(\omega ,{\mathbf{n}})|  \notag \\
& \hspace{5cm}\leq C^{2}\sum_{\mathbf{l}:\lambda _{\mathbf l}(\omega )\leq
\varepsilon _{F}}e^{-c|\mathbf{m-l}|-c|\mathbf{n-l}|}.  \label{eq:pqexp}
\end{align}%
Now the triangle inequality $|\mathbf{m-l}|+|\mathbf{n-l}|\geq |\mathbf{m-n}|$ yields the following upper bound for the r.h.s. \eqref{eq:pqexp}
\begin{equation}
C^{2}e^{-c|\mathbf{m-n}|/2}\sum_{\mathbf{l}\in \mathbb{Z}^{d}}e^{-c|\mathbf{
m-l}|/2}=\widetilde{C}e^{-\tilde{c}|\mathbf{m-n}|},  \label{eq:pqexp1}
\end{equation}
where
\begin{equation}
\widetilde{C}=C^{2}\sum_{\mathbf{m}\in \mathbb{Z}^{d}}e^{-c|\mathbf{m}%
|/2}<\infty .  \label{eq:pqexp2}
\end{equation}%
Combining \eqref{eq:pf_spec} -- \eqref{eq:pqexp2}, we obtain an analog of 
\eqref{eq:exp-decay} and \eqref{eq:ef-cor-decay} with some $\omega$-independent $c=\tilde{c},\,C=\tilde{C}$ for $P_{\omega }(\varepsilon _{F})$
and $Q_{I}$, but not for their expectations. Since, however, $c=\tilde{c}$
and $C=\tilde{C}$ are $\omega $-independent, the same bounds hold for the
corresponding expectations in Criterion~\ref{cr:1} and Criterion~\ref{cr:2},
which, in turn, imply the validity of the Area Law for the Maryland model.
\end{proof}

\smallskip

Here is one more example where the localisation centers are also explicit
and non-random as in the Maryland model.

\begin{proof}[Proof of Theorem~\protect\ref{th:ap-d}~$(ii)$]
Kachkovskiy, Parnovski and Shterenbereg consider in \cite{KPS} a class of $d$%
-dimensional Schr\"{o}dinger operators $\eqref{eq:schrodinger}$ with potentials \eqref{eq:lpp-d} described in
the assertion, i.e., with $\xi $-H\"{o}lder $1$-periodic monotone potentials $%
\eqref{eq:monotone-d1}$, and weakly 
Diophantine frequencies $\eqref{eq:weakly-Dio}$. 

An archetypal example in this class of operators is the multidimensional
Maryland model from Theorem \ref{th:ap-d}(i). However, unlike the above
class of operators, the Maryland model is explicitly solvable for any
non-zero coupling constants $g$ and for the Diophantine frequencies %
\eqref{eq:Dob}. 

The main (perturbative) result of \cite{KPS} states:

\begin{proposition}
\label{prop:monotone-d}[\cite{KPS}, Theorem~1.1] Let $\xi \geq 1,\, \rho,\mu
> 0,\, 0 < \delta < 1$. There exists $g_0 = g_0(d, \rho, \mu, \xi, \delta) >
0$ such that for every $g > g_0,\, \alpha\in\Omega_{\rho, \mu}$, and $\xi$%
-H\"older $(1$-periodic$)$ monotone function $v: \mathbb{R}\rightarrow [-\infty, +\infty)$
one can find a $1$-periodic function $\lambda: \mathbb{R}\rightarrow
[-\infty, +\infty)$, strictly increasing on $[0, 1)$, and a $1$-periodic
measurable function $\psi: \R \rightarrow \ell^2(\mathbb{Z}^d)$ such
that 
\begin{equation}  \label{eq:mar0}
H_\omega\psi(\omega) = \lambda(\omega)\psi(\omega),\, \text{for all}\,\
\omega\in\mathbb{R},
\end{equation}
and 
\begin{align*} 
& \|\psi(\omega)\|_{\ell^2(\mathbb{Z}^d)} = 1,\, |\psi(\omega, \mathbf{0}) -
1| < g^{-(1 - \delta)}  \notag \\
& \hspace{0.5cm} |\psi(\omega, \mathbf{n})| \leq g^{-(1 - \delta)|\mathbf{n}
|}\,\quad \text{for}\,\quad |\mathbf{n}|\neq 0,
\end{align*}
where $\psi(\omega, \mathbf{n}) = \langle \psi(\omega), \delta_{\mathbf{n}%
}\rangle$ denotes the components of the vector-valued function $\psi$ and $%
\{\delta_{\mathbf{n}}\}_{\mathbf{n }\in \mathbb{Z}^d}$ denotes the canonical
basis of $\ell^2(\mathbb{Z}^d)$.
\end{proposition}

It follows then from the proposition that

\begin{itemize}
\item 
the spectrum of $H_\omega$ occupies the whole $\mathbb{R}$,
 is pure point and of multiplicity 1, its eigenvalues 
 $\{\lambda_{\mathbf{l}}(\omega)\}_{\mathbf{l}\in
 \mathbb{Z}^d}$ and the corresponding eigenfunctions   
 $\{\psi_{\mathbf{l}}(\omega)\}_{\mathbf{l}\in
 \mathbb{Z}^d}$,  $\psi_{\mathbf{l}}(\omega)=
 \{\psi_{\mathbf{l}}(\omega,\mathbf{n})\}_{\mathbf{n}\in
 \mathbb{Z}^d}$ are
 \begin{equation}\label{eq:csr}
 \lambda_{\mathbf{l}}(\omega)=\lambda(\omega+ \langle \mathbf{\alpha},\mathbf{l}
 \rangle), \; \psi_{\mathbf{l}}(\omega,\mathbf{n})=\psi(\omega+ \langle \mathbf{\alpha},\mathbf{l}
 \rangle,\mathbf{n}-\mathbf{l}),
 \end{equation}


\smallskip

\item the eigenfunctions admit the bound $|\psi_{\mathbf{l}}(\omega,\mathbf{n%
})| \le Ce^{-c|\mathbf{n}-\mathbf{l}|}$ with the $\omega$-independent $C<
\infty$ and $c > 0$, i.e., we have the ULE property $\eqref{eq:ule}$.
\end{itemize}

\smallskip Thus, the description of the spectrum for this class of operators
coincides, at least for large $g$, with that of \eqref{eq:psize} -- %
\eqref{eq:maref-1} for the Maryland model and admits the exponential (ULE)
bound \eqref{eq:ule} analogous to \eqref{eq:th-ulema}. Thus, repeating
literally the proof of Theorem \ref{th:ap-d} (i), we arrive to the same
conclusion on the validity of the Area Law (and the exponential dynamical
localisation in expectation) on the whole spectrum of the operators in
question. 
\end{proof}

\begin{rmk}
The relations \eqref{eq:mar0} and \eqref{eq:csr} are known as the covariant
spectral representation and can also be easily checked in the Maryland
model. For their general discussion see \cite{GJLS}.
\end{rmk}

\begin{proof}[Proof of Theorem~\protect\ref{th:ap-d}~$(iii)$]
Ge, You and Zhou consider in \cite{GYZ} the quasi-periodic operators %
\eqref{eq:lro} -- \eqref{eq:W} 
with the multidimensional analog \eqref{eq:amo-d} of the almost Mathieu
potential $\eqref{eq:amo}.$ 
They establish 

\begin{proposition}
\label{prop:gyz}[\cite{GYZ}, Theorem~1.2] For $\alpha\in DC_d$ of %
\eqref{eq:DC_d} there exists $g_0(\alpha, d) > 0$, such that if $g > g_0$,
then the operator \eqref{eq:lro} -- \eqref{eq:W} with potential %
\eqref{eq:amo} exhibits exponential dynamical localisation in expectation $%
\eqref{eq:exp-dyn-loc}$. 
\end{proposition}

However, as follows from the analysis of the proof, the only property of the
evolution operator $e^{-itH}$ that is used in \cite{GYZ} is the bound $%
||e^{-itH}|| \leq 1$ that follows from the elementary bound $|e^{-itx}| \le
1 $ of the exponential function. Thus, replacing this function by a bounded
continuous $\phi$ of Definition \ref{def:ef-cor}, we obtain the bound $%
\eqref{eq:ef-cor-decay}$ instead of $\eqref{eq:exp-dyn-loc}$, and then
Criterion~\ref{cr:2} implies the assertion (iii) of Theorem.
\end{proof}

{
\begin{proof}[Proof of Theorem~\protect\ref{th:ap-d}~$(iv)$] We begin with two definitions (see \cite{A, G, Da-Ga:11, Da-Ga:13} for details).
\begin{defin}\label{def:cantor-group} $(i)$ $\Omega$ is a \textit{Cantor group} if it is an infinite, totally disconnected, compact Abelian topological group with no isolated points. In that case $\Omega$ has a unique translation invariant probability measure, Haar measure. One fixes a metric on $\Omega$ that is compatible with Haar measure.

$(ii)$ Consider a Cantor group $\Omega$ and a $\Z^d$ action by translations, $\{T^{\mathbf n}\}_{\mathbf{n}\in\Z^d}$. Namely, there are $\omega_1, \dots,\omega_d\in\Omega$ such that for every $\omega\in\Omega$
\[
T^{\mathbf n}\omega = \omega + \sum_{j = 1}^d n_j\omega_j,\quad \mathbf n = (n_1, \dots, n_d)\in\Z^d,
\]
where $+$ denotes the group operation. The action is called \textit{minimal} if all orbits are dense, namely if $\overline{\{T^{\mathbf n}\omega: \mathbf n\in\Z^d\}} = \Omega$ for every $\omega\in\Omega$. 
\end{defin}

Damanik and Gan consider in \cite{Da-Ga:13} a class of $d$-dimensional Schr\"odinger operators $\eqref{eq:schrodinger}$ with potential defined by $V_\omega(\mathbf n) = f(T^{\mathbf n}\omega)$ as in \eqref{eq:lp}, where $\omega\in\Omega$ is a Cantor group that admits a minimal $\Z^d$ action $T$ by translations, and $f\in C(\Omega,\R)$. The main result of \cite{Da-Ga:13} states:
\begin{proposition}\label{prop:da-ga}[\cite{Da-Ga:13}, Theorem~1.3] There exist a Cantor group $\Omega$ that admits a minimal $\Z^d$ action $T$ by translation, and an $f\in C(\Omega,\R)$ such that for every $\omega\in\Omega$ the Schr\"odinger operator \eqref{eq:schrodinger} with potential $f(T^{\mathbf n}\omega)$ of \eqref{eq:lp} has $\mathrm{ULE}$ \eqref{eq:eule} with $\omega$-independent constants.
\end{proposition}
\smallskip We conclude that the description of the spectrum for this class of operators coincides with that of \eqref{eq:psize} -- \eqref{eq:maref-1} for the Maryland model and admits the exponential (ULE) bound \eqref{eq:ule} analogous to \eqref{eq:th-ulema}. Thus, repeating literally the proof of Theorem \ref{th:ap-d} (i), we arrive to the same conclusion on the validity of the Area Law (and the exponential dynamical localisation in expectation) on the whole spectrum of the operators in question.
\end{proof}
}

\medskip 

\subsection{Proof of Theorem~\ref{th:ap-1d}}


We will now pass to the proof of Theorem~\ref{th:ap-1d} dealing
with one-dimensional almost periodic operators. 

\begin{proof}[Proof of Theorem~\protect\ref{th:ap-1d}~$(i)$]

{
The estimate \eqref{eq:ulema2}} allows us to apply again formulae \eqref{eq:pf_spec} -- \eqref{eq:pqexp2} to
obtain \eqref{eq:exp-decay} and then Criterion~\ref{cr:1} implies the
assertion of Theorem~\ref{th:ap-1d} (i).
\end{proof}

\begin{proof}[Proof of Theorem~\protect\ref{th:ap-1d}~$(ii)$] The Lipschitz monotone potentials 
$\eqref{eq:lpp} -\eqref{eq:lip}$ are similar to the potential $%
\eqref{eq:maryland}$ of Maryland model, however, are bounded. 
That is, the periodic sample function $v$ in $\eqref{eq:lpp}$ is strictly
increasing on $[0, 1)$, but has the jump discontinuities at integer
points. So the potential graph has a sawtooth shape.

Jitomirskaya and Kachkovskiy consider in \cite{JKa} the one-dimensional
discrete Schr\"{o}dinger operators with this class of potentials and
establish the following:

\begin{proposition}
\label{prop:monotone-1d}[\cite{JKa}, Corollary~3.5] Suppose $\alpha $ is
Diophantine, namely, there exist $C<\infty ,\tau >0$ such that 
$\Vert n\alpha \Vert >C|n|^{-\tau },
\;n\in \mathbb{N}$, 
where $\Vert x\Vert =\min (\{x\},\{1-x\})$. Let $\epsilon (\alpha
)=\liminf_{k}{q_{k-1}}/{q_{k+1}}$, where $\{q_{k}\}$ are the denominators of
the continued fraction approximants of $\alpha\, ($note that $\epsilon
(\alpha )\leq 1/2$ for any $\alpha \in \mathbb{R}\setminus \mathbb{Q})$.
Suppose that $g>{2e}/{((1-\epsilon (\alpha ))a_{-})}$.

Then, there exist $C<\infty ,c>0$ such that for any orthonormal
eigenfunction $\psi _{l}(\omega )$ there exists $n_{l}(\omega )$ such that
we have for $u_{l}(\omega )=\psi _{l}(\omega )/\psi _{l}(\omega ,0)$ 
\begin{equation}
|u_{l}(\omega ,n)|<Ce^{-c|n-n_{l}(\omega )|},  \label{eq:monotone-1d}
\end{equation}%
where $C<\infty $ and $c>0$ are $\omega $-independent.
\end{proposition}

This bound is similar to the one of \cite{JK}, but is valid for all $\omega$. 
However, unlike \eqref{eq:ulema2}, or, more generally, \eqref{eq:uule}, the
bound \eqref{eq:monotone-1d} cannot be used immediately to prove the
exponential decay of the Fermi projector and/or the eigenfunction
correlator, as was possible with the estimate \eqref{eq:th-ulema} and%
\eqref{eq:uule}, see formulae \eqref{eq:th-ulema}, \eqref{eq:pqexp2}. The
reason is that the functions $\{u_l\}_l$, although orthogonal, are not
orthonormalized. Therefore, an additional argument must be used to pass from 
$\{u_l\}_l$ to $\{\psi_l\}_l$ in the above bound. Such an elegant argument
was proposed in \cite{JK}. Thus, by using this argument and 
\eqref{eq:monotone-1d}, it is possible to prove the exponential dynamical
localisation in expectation as well as the exponential decay of the Fermi
projection, the eigenfunction correlator, and the assertion of Theorem~\ref{th:ap-1d} $(ii)$.
\end{proof}

\begin{proof}[Proof of Theorem~\protect\ref{th:ap-1d}~$(iii)$]
In the work \cite{JK} the exponential dynamical localisation in
expectation $\eqref{eq:exp-dyn-loc}$ was proved under the conditions of
Theorem \ref{th:ap-1d} $(iii)$. 
%
However, just like in \cite{GYZ} (see the proof of Theorem \ref{th:ap-d}
(iii)), the only property of the operator $e^{-itH}$ that was used in \cite%
{JK} was the bound $||e^{-itH}|| \leq 1$. Therefore, we repeat our argument
from the proof of Theorem~\ref{th:ap-d} (iii), allowing for the replacement
of $e^{-itH}$ by any $\phi(H) $ of Definition \ref{def:ef-cor}. We obtain a
more general bound $\eqref{eq:ef-cor-decay}$ instead of $%
\eqref{eq:exp-dyn-loc}$, and then Criterion~\ref{cr:2} implies the assertion
of Theorem~\ref{th:ap-1d} $(iii)$.
\end{proof}






\section{Proof of general criterion for ``strong" dynamical localisation}\label{sec:lem-pf}

{
Here we prove  Lemma~\ref{lem:bad-res}.
Let $G$ be the resolvent operator of $H$ and $G_{[m,p]}$ be the resolvent
operator of the restricted operator $H^{[m, p]}$. Let $n$ be such that $n =
l/20$ and assume that $k=0$. The second resolvent identity implies for each $%
i, j \in\{0,1\}$ 
\begin{equation*}
\begin{split}
G(0, l) =& G_{[-n + i, n - j]}(0, n - j)G(n - j + 1, l) \\
& +G_{[-n + i, n - j]}(0, -n + i)G(- n + i - 1, l),
\end{split}%
\end{equation*}
and 
\begin{equation*}
\begin{split}
G(0, l) = &G(0, l - n + i - 1)G_{[l - n + i, l + n - j]}(l - n + i, l) \\
& +G(0, l + n - j + 1)G_{[l - n + i, l + n - j]}(l + n - j, l).
\end{split}%
\end{equation*}
Thus, we obtain for $\lambda\notin \cap_{i, j \in\{0, 1\}}\mathrm{Bad}%
\,(H^{[- n + i, n - j]}, e^{-cn})$ for any $i, j \in\{0,1\}$ 
\begin{equation*}
|G(0, l)| \leq e^{-cn}\big(|G(n - j + 1, l)| + |G(- n + i - 1, l)|\big),
\end{equation*}
and for $\lambda\notin \cap_{i, j \in\{0, 1\}}\mathrm{Bad}\,(H^{[l - n + i,
l + n - j]}, e^{-cn})$ for any $i, j \in\{0,1\}$ 
\begin{equation*}
|G(0, l)| \leq e^{-cn}\big(|G(0, l - n + i - 1)| + |G(0, l + n - j + 1)|\big)%
.
\end{equation*}
Recall that \cite{AW} 
\begin{equation*}  
Q_I^\Lambda (m, n) = \lim_{\epsilon \rightarrow 0^+}\frac{\epsilon}{2}
\int_{I,\, I\subset \sigma(H)} |G_{\lambda + i0}^\Lambda (m, n)|^{1 -
\epsilon} \mathrm{d}\lambda \leq 1,
\end{equation*}
where $Q_I^\Lambda$ is the eigenfunction correlator of the operator $H$
restricted to a box $\Lambda\subset\mathbb{Z}$ and $G_{\lambda + i0}^\Lambda$
is its resolvent operator. Therefore, we get 
\begin{equation*}
\begin{split}
&\hspace{-0.5cm} 2\int_{I,\, I\subset \sigma(H)} \frac{\epsilon}{2}
|G_{\lambda + i0}(0, l)|^{1 - \epsilon} \mathrm{d}\lambda \\
& \leq 2\int_I e^{-(1 - \epsilon)cn} \frac{\epsilon}{2} \sum_{i, j \in
\{0,1\}} \{|G(n - j + 1, l)| + |G(- n + i - 1, l)| \\
& + |G(0, l - n + i - 1)| + |G(0, l + n - j + 1)|\}^{1 - \epsilon} \mathrm{d}%
\lambda.
\end{split}%
\end{equation*}
Since for any $\alpha \leq 1$ we have $|a + b|^\alpha \leq |a|^\alpha +
|b|^\alpha$, we have 
\begin{equation*}
\begin{split}
&2\int_{I,\, I \subset \sigma(H)} \frac{\epsilon}{2} |G_{\lambda + i0}(0,
l)|^{1 - \epsilon} \mathrm{d}\lambda \\
& \leq 2\int_I e^{-(1 - \epsilon)cn} \frac{\epsilon}{2} \sum_{i, j \in
\{0,1\}} \{|G(n - j + 1, l)|^{1 - \epsilon} + |G(- n + i - 1, l)|^{1 -
\epsilon} \\
&\hspace{2cm} + |G(0, l - n + i - 1)|^{1 - \epsilon} + |G(0, l + n - j +
1)|^{1 - \epsilon}\} \mathrm{d}\lambda.
\end{split}%
\end{equation*}
When we take $\epsilon\to 0$ and $\Lambda \nearrow \mathbb{Z}$ we obtain 
\begin{equation*}
\begin{split}
Q_I(0, l) &\leq 2 e^{-cn}\sum_{i, j \in \{0,1\}}\{Q_I(n - j + 1, l) + Q_I(-n
+ i -1, l) \\
& + Q_I(0, l - n + i - 1) + Q_I(0, l + n - j + 1)\} \\
& \leq 16 e^{-cn} \leq 16 e^{-\frac{c}{20} l} \equiv 16 e^{-c^{\prime }\, l},
\end{split}%
\end{equation*}
where in the second inequality we used the fact that $Q_I(n,m)\leq 1$ for
any $n, m\in\mathbb{Z}$ and the last one follows from the choice of $n$.
}

\section{Schr\"odinger operators with potentials generated by hyperbolic
transformations}

\label{sec:subshift}

In this section, we will present an extension of the results of works by {%
Avila, Damanik, and Zhang} in \cite{ADZ, ADZ1}, where the spectral
localisation was established. The proof of our main result is motivated by
and based on \cite[Theorem~2.10]{ADZ} and is an upgrade from spectral
localisation to the exponential decay of the corresponding eigenfunction
correlator (Theorem~\ref{th:subshift-ee}). This allows us to conclude that
the corresponding entanglement entropy $\eqref{eq:ee}$ obeys the Area Law at
the bottom of the spectrum (Theorem~\ref{subshit-ee}).

We will use two inputs from \cite{ADZ, ADZ1} -- the uniform positivity of
the Lyapunov exponent and a uniform large deviation type estimate.

\subsection{Necessary facts of the subshifts of finite type (SFT) and the SFT-driven Schr\"odinger operators}

\subsubsection{Shifts of finite type and corresponding Markov chains}


Let 
\begin{equation*}
\mathcal{A }= \{1, 2, \dots, k\},\quad k\geq 2,
\end{equation*}
be a finite alphabet equipped with the discrete topology. Consider the
product space $\mathcal{A}^\mathbb{Z}$ whose topology is generated by the
cylinder sets formed by fixing a finite set of coordinates 
\begin{equation*}
[n; a_0, a_1,\dots, a_m] = \{\omega\in\mathcal{A}^\mathbb{Z}\,|\, \omega_{n
+ i} = a_i,\, 0 \leq i \leq m\}.
\end{equation*}
The topology is metrizable and the metric is 
\begin{equation}  \label{eq:metric}
\mathrm{d}(\omega, \omega) = 0\, \text{for any}\, \omega\in\mathcal{A}^%
\mathbb{Z},
\end{equation}
\begin{equation*}
\mathrm{d}(\omega, \tilde\omega) = e^{-N(\omega,\tilde\omega)}\, \text{for
any}\, \omega, \tilde\omega\in\mathcal{A}^\mathbb{Z},\,
\omega\neq\tilde\omega,
\end{equation*}
where 
\begin{equation*}
N(\omega,\tilde\omega) = \max\,\{M\geq 0\, |\, \omega_n =
\tilde\omega_n\quad \text{for all}\,\, |n|\leq M\}.
\end{equation*}
Let $A=\{A_{ij}\}_{i,j=1}^k$ be a $k\times k$ matrix {such that} $%
A_{ij}\in\{ 0, 1\}$ for all $1\leq i, j \leq k$. Introduce the two-sided
shift of finite type as 
\begin{equation*}
\Omega_A = \{(\omega_n)_{n\in\mathbb{Z}}\in \mathcal{A}^\mathbb{Z}\, |\,
A_{\omega_n\omega_{n + 1}} = 1\quad \text{for all}\,\, n\in\mathbb{Z}\}.
\end{equation*}
The one-sided shift of finite type is defined in a similar way, just by
replacing $\mathbb{Z}$ by $\mathbb{Z}_+$.

The left shift map $T: \Omega_A \rightarrow \Omega_A$ is defined by $%
(T\omega)_n = \omega_{n + 1}$. The \textit{subshift of finite type} is a
restriction of $T$ to a closed \textit{shift-invariant} subspace $\Omega$,
namely $(\Omega, T)$ is a subshift of finite type if $\Omega\subset\Omega_A$
such that $T^n\Omega\subset\Omega$ for any $n\in\mathbb{Z}$.

Let $P=\{P_{ij}\}_{i,j=1}^k$ be a stochastic $k\times k$ matrix, i.e., $%
P_{ij} \geq 0$ for any $1 \leq i, j \leq k$ and $\sum_{i = 1}^k P_{ij} = 1$ for any $1 \leq j \leq k$.
We say that $P$ is compatible with the above matrix $A$ if $P_{ij} > 0
\Leftrightarrow A_{ij} = 1$. Assume that $P$ is irreducible, namely that for
all $i, j \in\{1, \dots, k\} = \mathcal{A}$ there exists an integer $n \geq
0 $ such that $(P^n)_{ij} > 0$. Since $P$ is an irreducible stochastic
matrix, the Perron-Frobenius Theorem states that there exists a unique
maximal eigenvalue of $P$, $\mathbf{\lambda} = 1$, and the rest of the
eigenvalues satisfy $|\lambda_j| < 1$. Let $\mathbf{p} = (p_1, \dots, p_k)$
be the eigenvector corresponding to the maximal eigenvalue $\lambda = 1$
satisfying $p_i > 0$ for all $1 \leq i \leq k$ normalized so that $\sum_{i =
1}^k p_i = 1$, such that $\mathbf{p}P = \mathbf{p}$. Given vector $\mathbf{p}
$, we can define a probability measure $\mathbb{P }= \mathbf{P}_P$ on $%
\Omega_A$ by 
\begin{equation}  \label{eq:markov-meas}
\mathbf{P}_P ([0; a_0, \dots, a_n]) = p_{a_0}P_{a_0a_1}P_{a_1a_2}\cdots
P_{a_{n - 1}a_n}
\end{equation}
on cylinder sets. By the Kolmogorov Extension Theorem, this uniquely defines
a measure on the whole $\sigma$-algebra. It is easy to check that the
measure $\mathbf{P}_P$ on $\Omega_A$ is $T$-invariant by checking that $%
\mathbf{P}_P$ and $T_*\mathbf{P}_P$ agree on cylinder sets where $T_*\mathbb{%
P}_P$ is the pushforward measure, namely $T_*\mathbf{P}_P(B) = \mathbf{P}%
_P(T^{-1}B)$ for any measurable $B\in\mathcal{A}^\mathbb{Z}$.

This measure is called a Markov measure, and it is well-known that the
topological support of $\mathbf{P}_P$ is a subshift of finite type $\Omega_A$
with the adjacency matrix $A$ where $A_{ij} = 1$ if and only if $P_{ij} > 0$%
. Moreover, the measure $\mathbf{P}_P$ is $T$-ergodic if and only if the
matrix $P$ is irreducible, and $T$ has a fixed point if and only if $P_{ii}
> 0$ for some $1 \leq i \leq k$, which implies that $P$ is aperiodic.

Let $(\Omega_A, T)$ be a subshift of finite type defined above equipped with
the $\sigma$-algebra, and let $\mathbf{P}$ be a probability measure on $%
\Omega_A$ that is ergodic with respect to the shift $T$. {We need the
following definition appearing in \cite{ADZ1, ADZ}:}

\begin{defin}
A positive measure $\nu$ possesses the bounded distortion property if there
exists a constant $C\geq 1$ such that for all cylinders $[n; a_0, \dots,
a_j] $ and $[m; b_0, \dots, b_l]$ in $\Omega_A$ where $m > n + j$ and $[n;
a_0, \dots, a_j] \cap [m; b_0, \dots, b_l] \neq \emptyset$, we have 
\begin{equation}  \label{eq:bds}
C^{-1} \leq \frac{\nu([n; a_0, \dots, a_j] \cap [m; b_0, \dots, b_l])}{%
\nu([n; a_0, \dots, a_j])\nu([m; b_0, \dots, b_l])} \leq C.
\end{equation}
\end{defin}

It follows from \cite[Lemma~3.4]{ADZ1} that the measure $\mathbf{P}_P$
defined by $\eqref{eq:markov-meas}$ has the bounded distortion property.


\subsubsection{Assumptions and definitions}


Here we formulate additional definitions and results from \cite{ADZ1, ADZ}
that we will need below.

The main object of our study in this section is the one-dimensional discrete
Schr\"odinger operators $H_\omega$ \eqref{eq:schrodinger1}, 
where the potential $V_\omega(n)$ is generated by the subshift of finite
type $(\Omega_A, T)$, namely $V_\omega(n) = v(T^n\omega), \omega\in\Omega_A$%
, where $v: \Omega_A \rightarrow \mathbb{R}$ is a function from one of the
following two classes (see \cite{ADZ1, ADZ}).

\begin{defin}
\label{def:LC-and-SH} $(i)$ A function $v: \Omega_A \rightarrow \mathbb{R}$
is said to be locally constant if there exists an integer $n_0 \geq 0$ such
that for each $\omega\in\Omega_A$, $v(\omega)$ depends only on the cylinder
set $[-n_0; \omega_{-n_0}, \dots \omega_{n_0}]$. Denote by $LC$ the set of
all locally constant functions $v:\Omega_A \rightarrow \mathbb{R}$.

$(ii)$ A function $v:\Omega_A \rightarrow \mathbb{R}$ is said to be $\alpha$%
-H\"older continuous for $0 < \alpha \leq 1$ if 
\begin{equation*}
\sup_{\omega\neq\tilde\omega}\frac{|v(\omega) - v(\tilde\omega)|}{\mathrm{d}%
(\omega, \tilde\omega)^\alpha} < \infty,
\end{equation*}
where $\mathrm{d}(\cdot, \cdot)$ is the metric on $\mathcal{A}^\mathbb{Z}$
defined by $\eqref{eq:metric}$. Denote by $C^\alpha(\Omega_A, \mathbb{R})$, $%
0 < \alpha \leq 1$, the space of real-valued $\alpha$-H\"older continuous
functions. The function $v\in C^\alpha(\Omega_A, \mathbb{R})$ is said to be
globally fiber bunched if there exists $\tau_0 > 0$ such that $\|v\|_\infty
< \tau_0$. Denote the set of all globally fiber bunched functions by $%
\mathrm{SH}$.
\end{defin}

Let 
\begin{equation*}  
A^\lambda(\omega)= \left( {%
\begin{array}{cc}
\lambda - v(\omega) & -1 \\ 
1 & 0 \\ 
\end{array}
} \right)
\end{equation*}
be the one-step transfer matrix corresponding to the Schr\"odinger operator $%
H_\omega$ of \eqref{eq:schrodinger1} with the potential generated by a
subshift of finite type such that $v\in \mathrm{LC}\cup\mathrm{SH}$ is
non-constant. For any $0\leq k \leq l$ let 
\begin{equation}  \label{eq:phi-k-l}
\Phi_{k, l}^\lambda(\omega) = A^\lambda(T^{l - 1}\omega)\cdots
A^\lambda(T^k\omega),
\end{equation}
be the corresponding $(l - k)$-step transfer matrix, where $T$ is given in %
\eqref{eq:ergodic-potential-1}. Denote by 
\begin{equation}  \label{eq:lyap-exp}
\gamma(\lambda) = \lim_{n\to\infty}\frac{1}{n}\mathbf{E}\{ \log \|\Phi_{0,
n}^\lambda(\omega)\| \ \}= \inf_{n \geq 1}\frac{1}{n} \log \|\Phi_{0,
n}^\lambda(\omega)\|,
\end{equation}
the Lyapunov exponent, corresponding to the operator $H_\omega$, where $%
\mathbf{E}\{\cdot\}$ is the expectation with respect to the $T$-ergodic
measure $\mathbf{P}$ associated with the subshift of finite type.

\smallskip

\noindent We have \cite{ADZ1, ADZ}:


\begin{defin}
\label{def:ple-uld} 
We say that

$(i)$ $A^\lambda$ has uniformly positive Lyapunov exponent $($PLE$)$ on $%
I\subset\mathbb{R}$ if 
\begin{equation}  \label{eq:ple}
\inf_{\lambda\in I}\gamma(\lambda) > 0.
\end{equation}
$(ii)$ $A^\lambda$ has uniform large deviation type estimate $($ULD$)$ on $%
I\subset\mathbb{R}$ if for every $\epsilon > 0$, there exist constants $%
C_{v, \epsilon}, c_{v, \epsilon} > 0$, depending only on $v$ and $\epsilon$,
such that for all $\lambda\in I$ and $n \geq 1$ 
\begin{equation}  \label{eq:uld}
\mathbf{P}\left\{\omega\in\mathbb{T}\, : \left| n^{-1} \log\|\Phi_{0,
n}^\lambda(\omega)\| - \gamma(\lambda)\right| > \epsilon\right\} < C_{v,
\epsilon}e^{-c_{v, \epsilon}n}.
\end{equation}
\end{defin}

Following the notation of \cite{ADZ1, ADZ}, we denote by $\mathcal{Z}_v =
\{\lambda\in \mathbb{R}: \gamma(\lambda) = 0\}$. Assume that $(\Omega_A, T)$
is a subshift of finite type and $\mathbf{P}$ is a $T$-ergodic measure that
has bounded distortion property $\eqref{eq:bds}$. Suppose that $T$ has a
fixed point on $\Omega_A$, and $v\in \mathrm{LC}\cup\mathrm{SH}$ is
non-constant. Then it follows from \cite[Theorem~1.3]{ADZ1} that $\mathcal{Z}%
_v$ is finite and there exists a set $\mathcal{F}_v \supset \mathcal{Z}_v$
such that for any compact interval $J$ and any $\eta > 0$, we have PLE $%
\eqref{eq:ple}$ on 
\begin{equation}  \label{eq:J-set}
J_\eta = J\setminus B(\mathcal{F}_v, \eta),
\end{equation}
where $B(\mathcal{F}_v, \eta)$ denotes the open $\eta$-neighborhood of $%
\mathcal{F}_v$. The set $J_\eta$ consists of a finite number of connected
compact intervals. Moreover, from \cite[Theorem~2.10]{ADZ} we conclude that
under the above conditions there exists a connected compact interval $J$
such that $\sigma(H_\omega)\subset J$, where $\sigma(H_\omega)$ is the
spectrum of the operator $H_\omega$, such that $A^\lambda$ satisfies ULD $%
\eqref{eq:uld}$ on $J_\eta$ for all $\eta > 0$.


\subsection{Proof of Theorems~\protect\ref{subshit-ee} and \protect\ref{th:subshift-ee}}

We start with several necessary assertions.

{
First, we prove the following large deviation type estimate for the Schr\"odinger
operators with potentials generated by the subshift of finite type under the
assumptions of Theorems~\ref{subshit-ee} and \ref{th:subshift-ee}.}


\begin{lemma}
\label{lem:prob-lde} For any $\lambda\in\mathbb{R}$ and any $\epsilon > 0$
we have 
\begin{align*}
\mathbf{P} \{\exists\, k, l,\, -n \leq k \leq l \leq n\, :\, | \log
\|\Phi_{k, l}^\lambda(\omega)\| - (l - k)\gamma(\lambda)| \geq \epsilon\,
n\} \leq C_{\epsilon, v} e^{-c_{\epsilon, v}\, n},
\end{align*}
where $C_{\epsilon, v} < \infty, c_{\epsilon, v} > 0$, and $\gamma(\cdot)$
is the corresponding Lyapunov exponent.
\end{lemma}

\begin{proof}
Let us choose $0 < \delta < 1$ such that 
\begin{equation*}
\delta \max (\log (2 + \|v\|_\infty), \gamma(\lambda)) \equiv \delta \max
(\log C_v, \gamma(\lambda)) \leq \frac{\epsilon}{2}.
\end{equation*}
We consider two cases: 
\begin{equation*}
(1)\, |l - k| \leq \delta n\,\quad \text{and}\quad (2)\, |l - k| \geq \delta
n.
\end{equation*}

First case: since for any $j$ we have $\|A^\lambda(T^j\omega)\| \leq C_v$,
we conclude that $\| \Phi_{k, l}^\lambda(\omega)\| \leq C_v^{|l - k|}$.
Therefore, since $|l - k|\leq \delta n$, we obtain 
\begin{equation}  \label{eq:small-case}
\log \| \Phi_{k, l}^\lambda(\omega)\| \leq |l - k| \log C_v \leq \delta n
\log C_v \leq \frac{\epsilon n}{2},
\end{equation}
where the last inequality follows from the choice of $\delta$.

Second case: it is proved in \cite[Theorem 2.10]{ADZ} that for $|l - k| \geq
\delta n$ we have 
\begin{equation*}  
\mathbf{P }\left\{ \omega\in\mathbb{T}\, :\, | \log \|\Phi_{k,
l}^\lambda(\omega)\| - (l - k)\gamma(\lambda)| > \frac{\epsilon\, n}{2}
\right\} \leq C_1 e^{-c_1\, n},
\end{equation*}
where $C_1 < \infty, c_1 > 0$ may depend only on $v$ and $\epsilon$, in
other words $A^\lambda$ has ULD on $I\subset\mathbb{R}$. By Borel-Cantelli
type argument we conclude that 
\begin{equation*}  
\begin{split}
&\mathbf{P}\left\{\exists\, k, l,\, -n \leq k \leq l \leq n,\, |l - k| \geq
\delta n :\, | \log \|\Phi_{k, l}^\lambda(\omega)\| - (l -
k)\gamma(\lambda)| \leq \frac{\epsilon\, n}{2} \right\} \\
& \geq 1 - \sum_{|l - k| \geq \delta n} C_1e^{-c_1\, n} \geq 1 - C_1\, 5n^2
e^{-c_1 n} \geq 1 - C_2 e^{-c_2 n},
\end{split}%
\end{equation*}
where $5n^2$ is the combinatorial factor bounding the number of ways to
choose $k, l$ from $[-n, n]\cap\mathbb{Z}$.

On the other hand, if $|l - k| \leq \delta n$, $\eqref{eq:small-case}$
implies 
\begin{equation*}
| \log \|\Phi_{k, l}^\lambda(\omega)\| - (l - k)\gamma(\lambda)| \leq \frac{%
\epsilon n}{2} + |l - k|\gamma(\lambda) \leq \frac{\epsilon n}{2} + \delta\,
n\,\gamma(\lambda) \leq \frac{\epsilon n}{2} + \frac{\epsilon n}{2} =
\epsilon n,
\end{equation*}
where the last inequality follows from the choice of $\delta$.
\end{proof}

For a given operator $H_\omega$ of the form $\eqref{eq:schrodinger1}$ it is
well-known that for any $a, b\in \mathbb{Z}$ the $|b - a|$-transfer matrix $%
\Phi_{a, b}^\lambda(\omega)$ of $\eqref{eq:phi-k-l}$ satisfies 
\begin{equation*}  
\Phi_{a, b}^\lambda(\omega) = \left( {%
\begin{array}{cc}
\mathrm{det} (H_\omega^{[a, b]} - \lambda) & \mathrm{det} (H_\omega^{[a + 1,
b]} - \lambda) \\ 
\mathrm{det} (H_\omega^{[a, b - 1]} - \lambda) & \mathrm{det} (H_\omega^{[a
+ 1, b - 1]} - \lambda) \\ 
\end{array}
} \right),
\end{equation*}
where $H_\omega^{[a, b]}$ is the restriction of the operator $H_\omega$ to $%
[a, b]\cap \mathbb{Z}$. Using Cramer's rule, we obtain the following
identity for its resolvent operator 
\begin{equation}  \label{eq:green-det}
G_{[a, b]}(k, l) = \frac{ \mathrm{det} (H_\omega^{[a, k - 1]} - \lambda) 
\mathrm{det} (H_\omega^{[l + 1, b]} - \lambda)}{ \mathrm{det} (H_\omega^{[a,
b]} - \lambda)},
\end{equation}
for any $a, b, k, l \in\mathbb{Z},\, a\leq k\leq l \leq b$.

Given a $(2n+1)\times(2n + 1)$ matrix $M$, we denote by 
\begin{equation}  \label{eq:res}
\mathrm{Res} (M, \epsilon) = \{ \lambda\in\mathbb{R}\, :\, \mathrm{dist}%
(\lambda, \sigma(M)) \leq \epsilon \},
\end{equation}
where $\sigma(M)$ denotes the spectrum of $M$, the set of spectral
parameters ``close" to the spectrum of $M$.

\begin{lemma}
\label{lem:lde} There exist $n_0 > 0$ and $c, \eta > 0$ such that for any $%
n\geq n_0$ 
\begin{equation}  \label{eq:lde}
\begin{split}
\mathbf{P}\, \{\exists i, j\in\{0, 1\}\, :\, \lambda\in&\mathrm{Bad}
(H_\omega^{[-n + i, n - j]}, e^{-cn})\cup \\
& \mathrm{Res} (H_\omega^{[-n + i, n - j]}, e^{-\frac{c}{10}n})\} \leq
e^{-\eta n},
\end{split}%
\end{equation}
where $\mathrm{Bad} (\cdot, \epsilon)$ is given by $\eqref{eq:bad}$ and $%
\mathrm{Res}(\cdot, \epsilon)$ is given by $\eqref{eq:res}$.
\end{lemma}

\begin{proof}
Assume that for all $k, l,\, -n + i \leq k \leq l \leq n - j,\ i,
j\in\{0,1\} $ 
\begin{equation}  \label{eq:phi-lem-lde}
| \log \|\Phi_{k, l}^\lambda(\omega)\| - (l - k)\gamma(\lambda)| \leq
\epsilon n.
\end{equation}
By Lemma~$\ref{lem:prob-lde}$, the inequality $\eqref{eq:phi-lem-lde}$ holds
always for $|l - k| \leq \delta n$ and for $|l - k| \geq \delta n$ it holds
with probability $> 1 - C_2 e^{-c_2 n}$. Since $\mathrm{det} (H_\omega^{[k,
l]} - \lambda)$ is an entry of the $|l - k|$- transfer matrix $\Phi_{k,
l}^\lambda(\omega)$ and $\|\Phi_{k, l}^\lambda(\omega)\|$ is its largest
eigenvalue, $\eqref{eq:phi-lem-lde}$ yields 
\begin{equation}  \label{eq:est-num}
\log |\mathrm{det} (H_\omega^{[k, l]} - \lambda)| \leq \log \|\Phi_{k,
l}^\lambda(\omega)\| \leq \epsilon n + |l - k|\gamma(\lambda).
\end{equation}
In addition, $\eqref{eq:phi-lem-lde}$ implies that 
\begin{equation}  \label{eq:est-denom-1}
(2n - i - j)\gamma(\lambda) - \log \|\Phi_{-n + i, n - j}^\lambda(\omega)\|
\leq \epsilon n.
\end{equation}
Since the norm of a $2\times 2$ matrix can be bounded by four times its
greatest entry, $\eqref{eq:est-denom-1}$ implies that 
\begin{equation}  \label{eq:est-denom-2}
(2n - i - j)\gamma(\lambda) - \max_{i, j\in\{0,1\}}\log |\mathrm{det}
(H_\omega^{[-n + i, n - j]} - \lambda)| \leq 2\epsilon n.
\end{equation}
Hence, for $i, j\in\{0,1\}$ satisfying $\eqref{eq:est-denom-2}$, using $%
\eqref{eq:est-num}$ and the representation $\eqref{eq:green-det}$ of the
corresponding resolvent operator, we obtain for any $-n + i \leq k \leq l
\leq n -j$ 
\begin{equation}  \label{eq:est-green}
\begin{split}
|G_{[-n + i, n - j]} (k, l)| &= \frac{|\mathrm{det} (H_\omega^{[-n + i, k -
1]} - \lambda)||\mathrm{det} (H_\omega^{[l + 1, n - j]} - \lambda)|}{|%
\mathrm{det} (H_\omega^{[-n + i, n - j]} - \lambda)|} \\
& \leq \frac{e^{\epsilon n + (k - 1 + n - i)\gamma(\lambda)}e^{\epsilon n +
(n - j - l - 1)\gamma(\lambda)}}{e^{(2n - i - j)\gamma(\lambda) - 2\epsilon
n}} \leq e^{4\epsilon n - \gamma(\lambda)(l - k + 2)},
\end{split}%
\end{equation}
where the last inequality holds because $i, j\in\{0, 1\}$. Since $%
\eqref{eq:est-green}$ holds for any $-n + i \leq k \leq l \leq n -j$, in
particular for $k = -n + i,\, l= 0$ and for $k = 0,\, l = n - j$, we obtain 
\begin{equation*}
|G_{[-n + i, n - j]} (0, -n + i)|, |G_{[-n + i, n - j]} (0, n - j)| \leq
e^{-n(\gamma(\lambda) - 4\epsilon)}e^{-\gamma(\lambda)}.
\end{equation*}
Since we have PLE on $I$, we conclude that there exist $c, \eta > 0$ such
that for any $n\geq n_0$ 
\begin{equation}  \label{eq:lde-bad}
\mathbf{P}\{\exists i, j\in\{0, 1\}\, :\, \lambda\in\mathrm{Bad}%
(H_\omega^{[-n + i, n - j]}, e^{-cn})\} \leq e^{-\eta n}.
\end{equation}
Since 
\begin{equation*}  
\|G_{[-n + i, n - j]}\| \leq \max_k \sum_l |G_{[-n + i, n - j]}(k, l)| \leq
Ce^{4\epsilon n},
\end{equation*}
where $C>0$ is some constant, we obtain 
\begin{equation*}
\mathrm{dist}(\lambda, \sigma(H_\omega^{[-n + i, n - j]})) \geq \frac{1}{C}%
e^{-4\epsilon n},
\end{equation*}
hence, $\lambda\notin \mathrm{Res}(H_\omega^{[-n + i, n - j]}, e^{-\frac{c}{%
10}n})$. As before, we can conclude that 
\begin{equation}  \label{eq:lde-res}
\mathbf{P}\{\exists i, j\in\{0, 1\}\, :\, \lambda\in\mathrm{Res}%
(H_\omega^{[-n + i, n - j]}, e^{-\frac{c}{10}n})\} \leq e^{-\eta n}.
\end{equation}
Combining $\eqref{eq:lde-bad}$ and $\eqref{eq:lde-res}$, we obtain $%
\eqref{eq:lde}$.
\end{proof}


Next we have the following general Lemma.

\begin{lemma}
\label{lem:bad} \label{lem:dist} Let $M, \widetilde M$ be $(2n + 1)\times(2n
+ 1)$ self-adjoint matrices, $\epsilon > 0$. If $\| M - \widetilde M\| \leq {%
\epsilon^2}/{100}$, then 
\begin{equation*}
\mathrm{Bad}\, (M, \epsilon) \subset \mathrm{Bad}\left(\widetilde M, \frac{%
\epsilon}{2} \right)\cup \mathrm{Res}\left(\widetilde M, \sqrt{\epsilon}
\right).
\end{equation*}
\end{lemma}

\begin{proof}
The second resolvent identity yields 
\begin{align*}  
&|(M - \lambda)^{-1}(0,n)|  \notag \\
& \hspace{1cm}\leq |(\widetilde M - \lambda)^{-1}(0,n)| + \|(M -
\lambda)^{-1}\|\|M - \widetilde M\|\|(\widetilde M - \lambda)^{-1}\|.
\end{align*}
If $\lambda\notin \mathrm{Bad}\left(\widetilde M, \frac{\epsilon}{2}
\right)\cup \mathrm{Res}\left(\widetilde M, \sqrt{\epsilon} \right)$, then 
\begin{equation*}
|(\widetilde M - \lambda)^{-1}(0,n)| < \frac{\epsilon}{2}\quad \text{and}%
\quad \|(\widetilde M - \lambda)^{-1}\|\leq \frac{1}{\sqrt\epsilon},
\end{equation*}
and since $\|\widetilde M - M\| \leq {\epsilon^2}/{100} $, we get 
\begin{equation*}
\|(M - \lambda)^{-1}\| \leq \frac{1}{\sqrt\epsilon - ({\epsilon^2}/{100})}
\leq \frac{2}{\sqrt\epsilon},
\end{equation*}
and, finally, 
\begin{align*}
&|(\widetilde M - \lambda)^{-1}(0,n)| + \|(M - \lambda)^{-1}\|\|M -
\widetilde M\|\|(\widetilde M - \lambda)^{-1}\| \\
& \hspace{6cm} \leq \frac{\epsilon}{2} + \frac{2}{\sqrt\epsilon}\frac{%
\epsilon^2}{100}\frac{1}{\sqrt\epsilon}\leq \epsilon.
\end{align*}
\end{proof}

Now we are ready to prove Theorems \ref{subshit-ee} and \ref{th:subshift-ee}.

\begin{proof}[Proof of Theorem~$\protect\ref{subshit-ee}$] 
\label{pr:th:doubling-map} 
{For a given $l\geq 100 n_0$ we fix $n = \lfloor l/100\rfloor$. Let $%
\omega\in\Omega_A$ be any element of the given subshift of finite type. Then
we define $\widetilde \omega\in\Omega_A$ as follows. Since the adjacency
matrix $A$ is irreducible and aperiodic, in every row of $A$ there is an
element equal to $1$, that is, for every $1\leq p\leq s$ there exists $1
\leq q(p) \leq s$ such that $A_{p, q(p)} = 1$. We define 
\begin{equation}  \label{eq:tilde-omega}
\widetilde\omega_m = \left\{ 
\begin{array}{lll}
\omega_m, \quad \quad\quad\,\, \text{if}\,\, |m| \leq n, &  &  \\ 
q(\omega_{m - 1}), \quad\, \text{if}\,\, m > n, &  &  \\ 
q(\omega_{m + 1}), \quad\, \text{if}\,\, m < - n, &  & 
\end{array}
\right.
\end{equation}
to have $\widetilde\omega$ the same as the given $\omega$ from $-n$ to $n$
and then continued in the way allowed by the corresponding adjacency matrix.
For a given function $v:\Omega_A\rightarrow\mathbb{R},\, v\in\mathrm{LC}\cup%
\mathrm{SH}$, we define 
\begin{equation*}
\widetilde v(\omega) = v(\widetilde\omega).
\end{equation*}
Now we can introduce a new potential that is an approximation of an original
potential 
\begin{equation*}
\widetilde V_\omega(k) = \widetilde v(T^{k} \omega),
\end{equation*}
and the new operator $\widetilde H_\omega = -\Delta + \widetilde V_\omega$.
Then we have for any $k\in\mathbb{Z}$ 
\begin{equation}  \label{eq:doubl-cut}
\begin{split}
|V_\omega(k) - \widetilde V_\omega(k)| = |v(T^k) - \widetilde
v(T^k\omega)|\leq C \mathrm{dist}(T^k\omega, \widetilde{T^k\omega}) \leq
e^{-cn},
\end{split}%
\end{equation}
where the last inequality follows from the definition of the metric $%
\eqref{eq:metric}$ taking into account that $\omega$ and $\widetilde\omega$
differ starting from $n$ by \eqref{eq:tilde-omega}. }

Note that if $v\in\mathrm{LC}$ then for $n$ sufficiently large we get $v =
\widetilde v$, thus the difference $\eqref{eq:doubl-cut}$ will be
identically $0$. If $v\in\mathrm{SH}$, then $c>0$ in the exponent contains $%
0 < \alpha \leq 1$. 

Take $\epsilon = e^{-cn/10}$, where $0 < c < {\log 2}/{2}$, is coming from
Lemma~$\ref{lem:lde}$. Then, since $\| H_\omega - \widetilde H_\omega\| \leq
e^{-cn} < \epsilon^2/100$, by applying Lemma~$\ref{lem:bad}$ twice, 
we obtain for any $i,j\in\{0,1\}$ 
\begin{equation*}
\begin{split}
&\mathrm{Bad}\, (H_\omega^{[k - n + i, k + n - j]}, \epsilon) \\
& \subset \mathrm{Bad}\, \left(\widetilde H_\omega^{[k - n + i, k + n - j]}, 
\frac{\epsilon}{2}\right)\cup \mathrm{Res}\, (\widetilde H_\omega^{[k - n +
i, k + n - j]}, \sqrt{\epsilon}) \\
& \subset \mathrm{Bad}\, \left(H_\omega^{[k - n + i, k + n - j]}, \frac{%
\epsilon}{4}\right)\cup \mathrm{Res}\, \left( H_\omega^{[k - n + i, k + n -
j]}, \sqrt{\frac{\epsilon}{2}}\right) \\
&\cup \mathrm{Res}\, (\widetilde H_\omega^{[k - n + i, k + n - j]}, \sqrt{%
\epsilon}).
\end{split}%
\end{equation*}
Using again $\|H_\omega - \widetilde H_\omega\|\leq {\epsilon^2}/{100}$, we
obtain for any $i,j\in\{0,1\}$ 
\begin{equation*}
\begin{split}
&\mathrm{Res}\, (\widetilde H_\omega^{[k - n + i, k + n - j]}, \sqrt{\epsilon%
})  \notag \\
& \subset \mathrm{Res}\, \left(H_\omega^{[k - n + i, k + n - j]}, \sqrt{%
\epsilon} + \frac{\epsilon^2}{100}\right)\subset (H_\omega^{[k - n + i, k +
n - j]}, 2\sqrt{\epsilon}).
\end{split}%
\end{equation*}
Thus, we obtain 
\begin{align}  \label{eq:prob-bad-res}
&\mathbf{P}\Big\{\exists i,j\in\{0,1\}\,:\, \lambda \in\mathrm{Bad}%
\,(\widetilde H_\omega^{[k - n + i, k +n - j]}, 4e^{-cn}) \\
& \hspace{0.5cm}\cup \mathrm{Res}\, (\widetilde H_\omega^{[k - n + i, k +n -
j]}, e^{-{cn}/{2}})\Big\} \leq \mathbf{P}\{\exists i,j\in\{0,1\}\,:\,  \notag
\\
& \hspace{1cm}\lambda\in\mathrm{Bad}\,(H_\omega^{[k - n + i, k +n - j]},
e^{-cn})\cup \mathrm{Res}\, (H_\omega^{[k - n + i, k +n - j]}, e^{-cn/10}
)\leq e^{-\eta n},  \notag
\end{align}
where the last inequality follows from Lemma~$\ref{lem:lde}$. The entries of
the resolvent corresponding to the finite operator $\widetilde H_\omega^{[k
- n + i, k +n - j]}$ are rational functions whose numerator and denominator
are polynomials of degree $\leq 3n$. Thus, for any $i,j\in\{0,1\} $ and any $%
k$ the set 
\begin{equation*}
\begin{split}
&\widetilde{\mathrm{Bad}}_k(\widetilde H_\omega^{[k - n + i, k +n - j]},
\epsilon) \\
& \equiv \cap_{i, j\in\{0,1 \}}\{\mathrm{Bad}\,\left(\widetilde H_\omega^{[k
- n + i, k +n - j]}, \frac{\epsilon}{2}\right)\cup\mathrm{Res}\, (\widetilde
H_\omega^{[k - n + i, k +n - j]}, \sqrt\epsilon)\}
\end{split}%
\end{equation*}
is a union of at most $48n$ intervals. 

Let us show that if $|k - l| \geq 100n$ then the sets $\widetilde{\mathrm{Bad%
}}_l (\widetilde H_\omega^{[k - n + i, k +n - j]}, \epsilon)$ and $%
\widetilde{\mathrm{Bad}}_l (\widetilde H_\omega^{[k - n + i, k +n - j]},
\epsilon)$ are ``almost independent", namely that there exists $\tilde{c} >
0 $ such that the following exponential mixing holds 
\begin{equation}  \label{eq:indepen}
\begin{split}
&|\mathbf{P} (\widetilde{\mathrm{Bad}}_l (\widetilde H_\omega^{[k - n + i, k
+n - j]}, \epsilon) \cap \widetilde{\mathrm{Bad}}_l (\widetilde H_\omega^{[l
- n + i, l +n - j]}, \epsilon)) - \\
& \mathbf{P}(\widetilde{\mathrm{Bad}}_l (\widetilde H_\omega^{[k - n + i, k
+n - j]}, \epsilon))\mathbf{P}(\widetilde{\mathrm{Bad}}_l (\widetilde
H_\omega^{[l - n + i, l +n - j]} \epsilon))| \leq e^{-\tilde c |k - l|} \leq
e^{-100\tilde c n}.
\end{split}%
\end{equation}


As we indicated in the preliminaries, there is a one-to-one correspondence
between subshifts of finite type and Markov chains. To prove $%
\eqref{eq:indepen}$, we will pass to the corresponding Markov chain setting.
Let $S = \{1, \dots, k\}$ be the phase space of the Markov chain. Since $%
\mathbf{P}$ is $T$-ergodic, we conclude that the corresponding transition
matrix $P$ is irreducible, hence, the Markov chain is irreducible. Since $T$
has a fixed point then $P$ and thus the Markov chain are aperiodic, namely
there exists an integer $m$ such that for all $1\leq i \leq k$ we have $%
(P^m)_{ii} > 0$. Thus, $\eqref{eq:indepen}$ follows from the \textit{%
Convergence Theorem} (see e.g. \cite[Theorem~4.9]{LPW}).

Assume that $|k - l| \geq 100 n$, $n = \lfloor\frac{l}{100}\rfloor$ and 
\begin{equation*}
\widetilde{\mathrm{Bad}}_k (\widetilde H_\omega^{[k - n + i, k +n - j]},
\epsilon) \cap \widetilde{\mathrm{Bad}}_l (\widetilde H_\omega^{[l - n + i,
l + n - j]}, \epsilon) \neq \emptyset.
\end{equation*}
Then, either one of the edges of set $\widetilde{\mathrm{Bad}}_k (\widetilde
H_\omega^{[k - n + i, k +n - j]}, \epsilon)$ is inside $\widetilde{\mathrm{%
Bad}}_l (\widetilde H_\omega^{[l - n + i, l +n - j]}, \epsilon)$ or vice
versa. 
Since there are at most $48n$ intervals in each of these sets, there are at
most $96n$ edges in total. Therefore we get 
\begin{equation}  \label{eq:bad-cap-est}
\begin{split}
&\mathbf{P }\{\big(\cap_{i, j\in\{0,1\}}\mathrm{Bad}\, (H_\omega^{[k -n + i,
k + n - j]}, 4e^{-cn})\big) \\
&\cap \big(\cap_{i, j \in\{0,1\}}\mathrm{Bad}\, (H_\omega^{[l -n + i, l + n
- j]}, 4e^{-cn})\big) \neq \emptyset\} \\
& \leq \mathbf{P}\{\widetilde{\mathrm{Bad}}_k (\widetilde H_\omega^{[k - n +
i, k +n - j]}, 2e^{-cn})\cap \widetilde{\mathrm{Bad}}_l (\widetilde
H_\omega^{[l - n + i, l +n - j]}, 2e^{-cn})\neq \emptyset\} \\
&\leq 192n e^{-\eta n} + e^{-\tilde C n},
\end{split}%
\end{equation}
where the last inequality follows from $\eqref{eq:indepen}$, $%
\eqref{eq:prob-bad-res}$, and the fact that the maximal number of edges in
both sets is $192n$. By Borel-Cantelli type argument we obtain 
\begin{equation*}
\begin{split}
\mathbf{P}\,\{\exists n_0: \, \forall n \geq n_0\,\, &\widetilde{\mathrm{Bad}%
}_k (\widetilde H_\omega^{[k - n + i, k +n - j]}, 2e^{-cn}) \\
& \cap \widetilde{\mathrm{Bad}}_l (\widetilde H_\omega^{[l - n + i, l +n -
j]}, 2e^{-cn}) = \emptyset\} \\
& \geq 1 - \sum_{n = 1}^\infty 192n e^{-n(\eta + \tilde C)} \geq 1 -
Ce^{-\eta n}.
\end{split}%
\end{equation*}
Now, applying Chebyshev's inequality, we conclude that for any $l\geq 100
n_0 $ 
\begin{equation*}
\begin{split}
\mathbf{E}\{Q_I(0, l) \} &\leq \mathbf{E}\{Q_I \mathbf{1}_{\widetilde{%
\mathrm{Bad}}_k (\widetilde H_\omega^{[k - n + i, k +n - j]}, 2e^{-cn})\cap 
\widetilde{\mathrm{Bad}}_l (\widetilde H_\omega^{[l - n + i, l +n - j]},
2e^{-cn}) = \emptyset} \} \\
& + \mathbf{E}\{Q_I (\mathbf{1} - \mathbf{1}_{\widetilde{\mathrm{Bad}}_k
(\widetilde H_\omega^{[k - n + i, k +n - j]}, 2e^{-cn})\cap \widetilde{%
\mathrm{Bad}}_l (\widetilde H_\omega^{[l - n + i, l +n - j]}, 2e^{-cn}) =
\emptyset} ) \} \\
& \leq 16 e^{-c^{\prime }n} + 192ne^{-(\eta + \tilde C) n} \leq C e^{-\tilde{%
c}l},
\end{split}%
\end{equation*}
where the second inequality follows from Lemma~$\ref{lem:bad-res}$ and $%
\eqref{eq:bad-cap-est}$.
\end{proof}

\begin{proof}[Proof of Theorem~\ref{th:subshift-ee}]
It follows from Theorem~\ref{subshit-ee} and Criterion~\ref{cr:2} that the corresponding entanglement entropy in this
case obeys the Area Law at the bottom of the spectrum.
\end{proof}

\begin{rmk}
An analogous proof can be given for half-line operators associated with
one-sided shift. In particular, our proof applies to the famous doubling map.
\end{rmk}


\appendix

\section{Another Proof of Bound (\protect\ref{eq:exp-decay}) for the
Maryland model.}

\label{ap:mapf}

We will consider a more general quantity $\mathbf{E}\{(\phi(H))(\mathbf{m},%
\mathbf{n})\}$ for a bounded $\phi:\mathbb{R}\rightarrow \mathbb{C}$. We
have by spectral theorem and \eqref{eq:psize} -- \eqref{eq:maref-1}%
\begin{align*}
&
\mathbf{E}\{|(\phi(H))(\mathbf{m},\mathbf{n})|\} \leq \sum_{\mathbf{l}\in \mathbb{Z}^{d}}\mathbf{E}\{|\phi(\lambda _{\mathbf{l}}(\omega))\psi_{\mathbf{%
l}}(\omega ,\mathbf{m)}\psi_{\mathbf{l}}(\omega, \mathbf{n})|\} \\
& 
=\sum_{\mathbf{l}\in \mathbb{Z}^{d}}\int_{\mathbb{T}%
^{d}}|\phi(\lambda _{0}(\omega +\langle \mathbf{\alpha },\mathbf{l}\rangle
)| |\psi _{0}(\lambda _{0}(\omega +\langle \mathbf{\alpha },\mathbf{l}%
\rangle ),\mathbf{m-l)}  \psi _{0}(\lambda _{0}(\omega +\langle \mathbf{\alpha },\mathbf{l}\rangle ),\mathbf{n}-\mathbf{l})\mathbf{)|}d\omega \\
& \hspace{4cm}=\sum_{\mathbf{l}\in \mathbb{Z}^{d}}\int_{\mathbb{T}%
^{d}}|\phi(\lambda _{0}(\omega )|\,|\psi _{0}(\lambda _{0}(\omega ),\mathbf{%
m-l)}\psi_0(\lambda _{0}(\omega ),\mathbf{n-l)|}d\omega ,
\end{align*}%
and in the last equality we use the periodicity of $\lambda _{0}:\mathbb{T}%
\rightarrow \mathbb{R}$.

Next, it follows from \eqref{eq:ncm} --\ \eqref{eq:laN} that $d\omega
=n(\lambda )d\lambda $, hence, the r.h.s. above is%
\begin{equation*}
\int_{\mathbb{R}}|\phi(\lambda )|\sum_{\mathbf{l}\in \mathbb{Z}^{d}}|\psi
_{0}(\lambda ,\mathbf{m-l)}\psi _{0}(\lambda ,\mathbf{n-l)|}n(\lambda
)d\lambda .
\end{equation*}%
Now, by using \eqref{eq:th-ulema} and simple formulae \eqref{eq:pf_spec} -- \eqref{eq:pqexp2}, we obtain 
\begin{equation*}
\mathbf{E}\{|(\phi(H))(\mathbf{m},\mathbf{n})|\}\leq \widetilde{C}e^{-c|%
\mathbf{m-n}|/2}\int_{\mathbb{R}}|\phi(\lambda )|n(\lambda )d\lambda .
\end{equation*}%
In particular, using the indicator $\chi _{(-\infty ,\varepsilon _{F}]}$ as $%
\phi$, we obtain having \eqref{eq:ncm} (cf. \eqref{eq:th-ulema})%
\begin{equation*}
\mathbf{E}\{|P(\mathbf{m},\mathbf{n})|\}|\leq \widetilde{C}e^{-\widetilde{c}|%
\mathbf{m-n}|} N(\varepsilon _{F}).
\end{equation*}
The role of $C < \infty$ in \eqref{eq:exp-decay} plays now $\widetilde C
N(\varepsilon_F) < C$ and $0 < c = 2\widetilde c$.



\section*{Acknowledgements}

The authors would like to thank S.~Jitomirskaya and S.~Sodin for useful
discussions. We would also like to thank A.~Elgart, I.~Kachkovskiy, and P.~M\"uller for useful comments on the preliminary version of the paper.

\end{document}